\documentclass[final,3p,12pt,sort&compress]{elsarticle}
\makeatletter
\def\ps@pprintTitle{%
 \let\@oddhead\@empty
 \let\@evenhead\@empty
 \def\@oddfoot{\centerline{\thepage}}%
 \let\@evenfoot\@oddfoot}
\makeatother
\usepackage{comment}
\usepackage{amssymb,amsmath, amsthm}
\usepackage{slashed}

\usepackage{epigraph}
\epigraphsize{\small\itshape}
\usepackage[bookmarks]{hyperref}
\theoremstyle{plain}
\newtheorem{thm}{Theorem}[section]
\newtheorem{lem}[thm]{Lemma}
\newtheorem{prop}[thm]{Proposition}
\newtheorem{rmk}[thm]{Note}
\newtheorem*{rmk*}{Note}
\newtheorem*{cor}{Corollary}
\theoremstyle{definition}
\newtheorem{defn}[thm]{Definition}
\newtheorem*{defn*}{Definition}

\theoremstyle{remark}

\makeatletter
\renewcommand{\@epitext}[1]{
\itshape \begin{minipage}{\epigraphwidth}\begin{\textflush} #1
\end{\textflush}\end{minipage}\vspace{1ex}}
\makeatother

\begin{document}

\title{On the real representations of the Poincare group}
\author{Leonardo Pedro}
\address{Centro de Fisica Teorica de Particulas, CFTP, 
Departamento de Fisica, Instituto Superior Tecnico, Universidade
  Tecnica de Lisboa}
\date{\today}

\begin{abstract}
The formulation of quantum mechanics 
with a complex Hilbert space is equivalent to a formulation with a real Hilbert space and 
particular density matrix and observables. We study the real representations of the
Poincare group, motivated by the fact that the localization of complex unitary representations of the Poincare
group is incompatible with causality, Poincare covariance and
energy positivity.

We review the map from the complex to the real irreducible
representations---finite-dimensional or unitary---of a Lie group on a
Hilbert space. Then we show that all the finite-dimensional real
representations of the identity component of the Lorentz group are also
representations of the parity, in contrast with many complex representations.

We show that any localizable unitary representation of the Poincare group, 
compatible with Poincare covariance, verifies: 
1) it is a direct sum of irreducible representations which are massive or massless with discrete helicity.
2) it respects causality;
3) if it is complex it contains necessarily both positive and negative energy subrepresentations
4) it is an irreducible representation of the Poincare group (including parity) if and only if it is: a)real and b)massive with spin 1/2 or 
massless with helicity 1/2. Finally, the energy positivity problem is discussed in a many-particles context. 
\end{abstract}

\maketitle

\begin{epigraphs}
\qitem{A state \emph{[of a spin-0 elementary system]} which is localized at the origin in 
one coordinate system, is not localized in a moving coordinate system, 
even if the origins coincide at t=0.
Hence our \emph{[position]} operators have no simple
covariant meaning under relativistic transformations.[...]

For higher but finite \emph{[spin of a massless representation]} s, beginning with s=1
(i.e. Maxwell's equations) we found that no localized states in the
above sense exist. This is an unsatisfactory, if not unexpected,
feature of our work.}
{---\textup{E.P.Wigner \& T.D.Newton (1949)\cite{newton}}}
\qitem{The concepts
of mathematics are not chosen for their conceptual simplicity---even
sequences of pairs of numbers \emph{[i.e. the real numbers]} are far from being
the simplest concepts---but for their amenability to clever
manipulations and to striking, brilliant arguments. Let us not forget
that the Hilbert space of quantum mechanics is the complex Hilbert
space, with a Hermitean scalar product. Surely
to the unpreoccupied mind, complex numbers are far from natural or simple
and they cannot be suggested by physical observations. Furthermore, the
use of complex numbers is in this case not a calculational trick of applied
mathematics but comes close to being a necessity in the formulation of
the laws of quantum mechanics.}
{---\textup{E.P.Wigner (1959)\cite{wignerquote}}}
\end{epigraphs}

\section{Introduction}

\subsection{Motivation}
Henri Poincar{\'e} defined the Poincare
group as the set of transformations
that leave invariant the Maxwell equations for the classical
electromagnetic field. The classical electromagnetic field transforms as a real 
representation of the Poincare group.

The complex representations of the Poincare group were systematically
studied\cite{wigner,mackey,ohnuki,poincare,weinberg,knapp} 
and used in the definition of quantum fields\cite{feynmanrules}. 
These studies were very important in the evolution of the
role of symmetry in the Quantum Theory\cite{symmetry}.

The formulation of quantum mechanics 
with a complex Hilbert space is equivalent to a formulation with a real Hilbert space and 
particular density matrix and observables\cite{realQM}. 
Quantum Theory on real Hilbert spaces
was investigated before\cite{realqft,realqftII,realqftIII,quantumstatistics,hestenes_old}, the main conclusion was that 
the formulation of non-relativistic Quantum Mechanics with a real Hilbert space is 
necessarily equivalent to the complex formulation. 
We could not find in the literature a systematic study on the real
representations of the Poincare group, as it seems to be 
common assumptions that if non-relativistic Quantum Mechanics is necessarily 
complex then the relativistic version must also be---it is hard to accept this specially 
because a relativistic Quantum Theory for a single-particle is inconsistent,
as relativistic causality requires the existence of anti-particles\cite{weinberg}---or 
that the energy positivity implies complex Poincare representations---it is a long shot, 
as it happens for the relativistic causality, only in a many-particles description 
the energy positivity is well defined.

The reasons motivating this study are:

1) The real representations of the Poincare group play an
important role in the classical electromagnetism and general
relativity\cite{classicalfields,gravitypoincare,gravitypoincare2} and in Quantum Theory---
e.g. the Higgs boson, Majorana fermion or quantum electromagnetic fields transform as real
representations under the action of the Poincare group.

2) The parity---included in the full Poincare group---and
charge-parity transformations are not symmetries of the Electroweak
interactions\cite{brancocp}. It is not clear why the charge-parity is
an apparent symmetry of the Strong interactions\cite{strongcp} or how
to explain the matter-antimatter asymmetry\cite{imbalance} through the
charge-parity violation. Since the self-conjugate finite-dimensional 
representations of the identity component of the Lorentz group are also 
representations of the parity, this work may be useful in future studies
of the parity and charge-parity violations.

3) The localization of complex irreducible unitary representations of the Poincare
group is incompatible with causality, Poincare covariance and
energy positivity\cite{localization,causality,stringfields}, while the
complex representation corresponding to the photon is not
localizable\cite{newton,wightman,vara}. The localization problems in the complex 
representations may come from the representation of the charge and matter-antimatter 
properties in relativistic Quantum Mechanics---which has always been problematic, 
remember the Dirac sea\cite{diracsea}---and so a study of the real representations, 
necessarily independent of the charge and matter-antimatter properties, may be useful.

\subsection{Systems on real and complex Hilbert spaces}
The position operator in Quantum Mechanics is mathematically expressed using
a system of imprimitivity: a set of projection operators---
associated with the coordinate space---on a Hilbert space;
a group acting both on the Hilbert space and on the coordinate space 
in a consistent way\cite{vara,mathQM}.

Many representations of a group---such as the finite-dimensional
representations of semisimple Lie  groups\cite{Hall} or the unitary
representations of separable locally compact
groups\cite{locallycompact}---are direct sums (or integrals) of
irreducible representations, hence the study of these representations 
reduces to the study of the irreducible representations.

If the set of normal operators commuting with an irreducible real unitary 
representation of the Poincare group is isomorphic to the quaternions or to the complex numbers,
then the most general position operator that the representation space admits is not complex linear,
but real linear. Therefore, in this case, the real irreducible representations 
generalize the complex ones and these in turn generalize the quaternionic ones.

The study of irreducible representations on complex Hilbert
spaces is in general easier than on real Hilbert spaces, because the
field of complex numbers is the algebraic closure --- where any
polynomial equation has a root --- of the field of real numbers. 
There is a well studied map, one-to-one or two-to-one and surjective
up to equivalence, from the complex to the real linear
finite-dimensional irreducible representations of a real Lie
algebra\cite{realalgebras,realirrep}.

Section \ref{section:Systems} reviews a similar map from the complex
to the real irreducible representations---finite-dimensional or
unitary---of a Lie group on a Hilbert space. 
Using Mackey's imprimitivity theorem, we extend the map to systems of 
imprimitivity. This section follows closely the reference\cite{realalgebras}, with the 
addition that we will also use the Schur's lemma for unitary representations
on a complex Hilbert space\cite{schur}.

Related studies can be found in the references
\cite{compactlie,spinorsrealhilbert}.

\subsection{Finite-dimensional representations of the Lorentz group}

The Poincare group, also called inhomogeneous Lorentz group, is the
semi-direct product of the translations and Lorentz Lie
groups\cite{Hall}. Whether or not the Lorentz and Poincare groups
include the parity and time reversal transformations depends on the
context and authors. To be clear, we use the prefixes full/restricted
when including/excluding parity and time reversal transformations. The
Pin(3,1)/SL(2,C) groups are double covers of the full/restricted
Lorentz group. The semi-direct product of the translations with the
Pin(3,1)/SL(2,C) groups is called IPin(3,1)/ISL(2,C) Lie group --- the
letter (I) stands for inhomogeneous.

A projective representation of the Poincare group on a complex/real
Hilbert space is an homomorphism, defined up to a complex phase/sign,
from the group to the automorphisms of the Hilbert space. Since the
IPin(3,1) group is a double cover of the full Poincare group, their
projective representations are the same\cite{pin}. 
All finite-dimensional projective representations of a simply
connected group, such as SL(2,C), are usual
representations\cite{weinberg}.
Both SL(2,C) and Pin(3,1) are semi-simple Lie groups, and so all its
finite-dimensional representations are direct sums of irreducible
representations\cite{Hall}. Therefore, the study of the
finite-dimensional projective representations of the restricted
Lorentz group reduces to the study of the finite-dimensional
irreducible representations of SL(2,C).

The Dirac spinor is an element of a 4 dimensional complex vector
space, while the Majorana spinor is an element of a 4 dimensional real
vector space\cite{todorov,irreducible, pal, dreiner}. The complex
finite-dimensional irreducible representations of SL(2,C) can be
written as linear combinations of tensor products of Dirac spinors.

In Section \ref{section:Finite} we will review the Pin(3,1) and
SL(2,C) semi-simple Lie groups and its relation with the Majorana,
Dirac and Pauli matrices. We will obtain all the real
finite-dimensional irreducible representations of SL(2,C) as linear
combinations of tensor products of Majorana spinors, using the map
from Section \ref{section:Systems}.
Then we will check that all these real representations are also
projective representations of the full Lorentz group, in contrast with
the complex representations which are not all projective
representations of the full Lorentz group.

\subsection{Unitary representations of the Poincare group}

According to Wigner's theorem, the most general transformations,
leaving invariant the modulus of the internal product of a 
Hilbert space, are: unitary or anti-unitary operators, defined up to a
complex phase, for a complex Hilbert; unitary, defined
up to a signal, for a real Hilbert\cite{wignertheorem, vara}. This motivates
the study of the (anti-)unitary projective representations of the full
Poincare group.

All (anti-)unitary projective representations of ISL(2,C) are, up to
isomorphisms, well defined unitary representations, because ISL(2,C)
is simply connected\cite{weinberg}. Both ISL(2,C) and IPin(3,1) are
separable locally compact groups and so all its (anti-)unitary
projective representations are direct integrals of irreducible
representations\cite{locallycompact}. Therefore, the study of the
(anti-)unitary projective representations of the restricted Poincare
group reduces to the study of the unitary irreducible representations
of ISL(2,C).

The spinor fields, space-time dependent spinors, are solutions of the
free Dirac equation\cite{Dirac}. The real/complex Bargmann-Wigner
fields\cite{BW,allspins}, space-time dependent linear combinations of
tensor products of Majorana/Dirac spinors, are solutions of the free
Dirac equation in each tensor index. The complex unitary irreducible
projective representations of the Poincare group with discrete spin or
helicity can be written as complex Bargmann-Wigner fields.

In Section \ref{section:Unitary}, we will obtain all the real
unitary irreducible projective representations of the Poincare
group, with discrete spin or helicity, as real Bargmann-Wigner fields, using the
map from Section 2. For each pair of complex representations with
positive/negative energy, there is one real representation.
We will define the Majorana-Fourier and Majorana-Hankel unitary
transforms of the real Bargmann-Wigner fields, 
relating the coordinate space with the linear and angular momenta
spaces.
We show that any localizable unitary representation of the Poincare group, 
compatible with Poincare covariance, verifies: 
1) it is a direct sum of irreducible representations which are massive or massless with discrete helicity.
2) it respects causality;
3) if it is complex it contains necessarily both positive and negative energy subrepresentations
4) it is an irreducible representation of the Poincare group (including parity) if and only if it is: a)real and b)massive with spin 1/2 or 
massless with helicity 1/2. 

The free Dirac equation is diagonal in the Newton-Wigner
representation\cite{newton}, related to the Dirac representation
through a Foldy-Wouthuysen transformation\cite{revfoldy,foldy} of
Dirac spinor fields. The Majorana-Fourier transform, when applied on
Dirac spinor fields, is related with the Newton-Wigner representation
and the Foldy-Wouthuysen transformation. In the context of Clifford
Algebras, there are studies on the geometric square roots of -1 
\cite{hestenes_old,squareroot} and on the
generalizations of the Fourier transform\cite{clifford}, with
applications to image processing\cite{image}.

\subsection{Energy Positivity}
In non-relativistic Quantum Mechanics the time is invariant under the
Galilean transformations ---excluding the time reversal transformation---and so 
the generator of translations in time is also invariant. 
Therefore, the positivity of the Energy and the localization in space of a state
can be defined simultaneously.
In relativistic Quantum Mechanics, the time is not invariant under Lorentz transformations, as a consequence 
the positivity of the Energy and the localization in space of a state cannot be defined simultaneously---
the corresponding projection operators do not commute. The solution can be found in a many particles system. 

In the canonical quantization description of a many particles system, the positivity of Energy is well defined
by construction and the localization problem is handled by introducing anti-particles---causality implies the
existence of anti-particles\cite{weinberg}, a related approach led Dirac to predict the positron\cite{diracsea}.
Yet, it should also be possible to build a description of a many particles system where the localization in space of a state 
is well defined by construction and the Energy positivity problem can be handled, as we can infer from the canonical quantization 
that both Energy positivity and localization are important and, in some way, complementary.
Dirac himself was the first to consider an approach which do not assume the positivity of Energy by construction\cite{negativeprobs} 
and quantization in de Sitter space-time may be achieved in a related approach\cite{desitter,krein}.

The description of a many-particles system based on the localization will be discussed in the section \ref{section:Energy}.

\section{Systems on real and complex Hilbert spaces}
\label{section:Systems}

\begin{defn*}[System]
A system $(M,V)$ is defined by:\\
1) the (real or complex) Hilbert space $V$;\\
2) a set $M$ of bounded endomorphisms on $V$.
\end{defn*}

The representation of a symmetry is an example of a system: a representation space plus a set of operators 
representing the action of the symmetry group in the representation space\cite{representations}. 

\begin{defn*}[Complexification]
Consider a system $(M,W)$ on a
real Hilbert space. The system $(M,W^c)$ is the complexification of the
system $(M,W)$, defined as 
$W^c\equiv  \mathbb{C}\otimes W$, with the multiplication by scalars
such that $a(b w)\equiv (ab)w$
for $a,b\in \mathbb{C}$ and $w\in W$.
The internal product of $W^c$ is defined---for $u_r,u_i,v_r,v_i\in W$ and $<v_r,u_r>$ the internal product of $W$---as:
\begin{align*}
<v_r+i v_i,u_r+i u_i>_c\equiv
<v_r,u_r>+<v_i,u_i>+i<v_r,u_i>-i<v_i,u_r>
\end{align*}
\end{defn*}

\begin{defn*}[Realification]
Consider a system $(M,V)$ on a
complex Hilbert space. 
The system$(M,V^r)$ is the realification of the
system $(M,V)$, defined as 
$V^r\equiv V$ is a real Hilbert space 
with the multiplication by scalars restricted to reals
such that $a(v)\equiv (a+i0)v$ 
for $a\in \mathbb{R}$ and $v\in V$. 
The internal product of $V^r$ is defined---for $u,v\in V$ and $<v,u>$ is the internal product of $V$---as:
\begin{align*}
<v,u>_r\equiv \frac{<v,u>+<u,v>}{2}
\end{align*}
\end{defn*}

\begin{rmk}
Let $H_n$, with $n\in\{1,2\}$, be two Hilbert spaces with internal
products $<,>:H_n\times H_n\to
\mathbb{F}$,($\mathbb{F}=\mathbb{R},\mathbb{C}$).
A (anti-)linear operator $U:H_1\to H_2$ is (anti-)unitary iff:\\
1) it is surjective;\\
2) for all $x\in H_1$, $<U(x) , U(x)>=<x, x>$.
\end{rmk}

\begin{prop}
\label{prop:unitary}
Let $H_n$, with $n\in\{1,2\}$, be two complex Hilbert spaces and
$H^r_n$ its complexification.
The following two statements are equivalent:

1) The operator $U:H_1\to H_2$ is (anti-)unitary;

2) The operator $U^r:H_1^r\to H_2^r$ is (anti-)unitary, where 
$U^r(h)\equiv U(h)$, for $h\in H_1$.
\end{prop}
\begin{proof}
Since $<h,h>=<h,h>_r$ and 
$U^r(h)=U(h)$, for $h\in H_1$, we get the result.
\end{proof}

\begin{defn*}[Equivalence]
Consider the systems $(M,V)$ and $(N,W)$:\\
1) A normal endomorphism of $(M,V)$ is a bounded endomorphism $S:V\to V$
commuting with $S^\dagger$ and $m$, for all $m\in M$; an anti-endomorphism in a 
complex Hilbert space is an anti-linear endomorphism;\\
2) An isometry of $(M,V)$ is a unitary operator $S:V\to V$
commuting with $m$, for all $m\in M$;\\
3) The systems $(M,V)$ and $(N,W)$ are unitary equivalent iff there is a isometry
$\alpha:V\to W$ such that $N=\{\alpha m\alpha^{\dagger}: m\in M\}$.\\
\end{defn*}


We use the trivial extension of the definition of irreducibility from representations to systems.

\begin{defn*}[Irreducibility]
Consider the system $(M,V)$ and let $W$ be a linear subspace of $V$:\\
1) $(M,W)$ is a (topological) subsystem of $(M,V)$
iff $W$ is closed and invariant under the system action, that is, for all $w\in W$:$(m w)\in W$, for all $m\in M$;\\
2) A system $(M,V)$ is (topologically) irreducible iff their only 
sub-systems are the non-proper $(M,V)$ or trivial  $(M,\{0\})$ sub-systems, 
where $\{0\}$ is the null space.
\end{defn*}

\begin{defn*}[Structures]
1) Consider a system $(M,V)$ on a
complex Hilbert space. 
A C-conjugation operator of $(M,V)$ is an 
anti-unitary involution of $V$ commuting with $m$,
for all $m\in M$;\\
2) Consider a system $(M,W)$ on a
real Hilbert space. 
A R-imaginary operator of $(M,W)$, $J$, is an isometry of
$(M,W)$ verifying $J^2=-1$. 
\end{defn*}

\subsection{The map from the complex to the real systems}

\begin{defn*}
Consider an irreducible system $(M,V)$ on a
complex Hilbert space:\\
1) The system is C-real iff there is a C-conjugation
operator;\\
2) The system is C-pseudoreal iff 
there is no C-conjugation operator but there is an
anti-unitary operator of $(M,V)$;\\
3) The system is C-complex iff 
there is no anti-unitary operator of $(M,V)$.
\end{defn*}

\begin{defn}
\label{defn:Rsystem}
Consider the system $(M,W)$ on a
real Hilbert space and let $(M,W^c)$ be its complexification:
1) $(M,W)$ is R-real iff $(M,W^c)$ is C-real irreducible;\\
2) $(M,W)$ is R-pseudoreal iff $(M,V)$ is C-pseudoreal irreducible,
with $W^c=V\oplus \bar V$;
3) $(M,W)$ is R-complex iff $(M,V)$ is C-complex irreducible, with
$W^c=V\oplus \bar V$.
\end{defn}

\begin{prop}
Any irreducible real system is R-real or R-pseudoreal or R-complex.
\end{prop}
\begin{proof}
Consider an irreducible system $(M,W)$
on a real Hilbert space.
There is a C-conjugation operator of $(M,W^c)$, 
$\theta$, defined by $\theta(u+iv)\equiv (u-iv)$ for $u,v\in W$,
verifying $(W^c)_\theta=W$.

Let $(M,X^c)$ be a proper non-trivial subsystem of
$(M,W^c)$. Then $\theta$ is a C-conjugation operator of the 
subsystems $(M,Y^c)$ and $(M,Z^c)$, where 
$Y^c\equiv \{u+\theta v: u,v\in X^c\}$ and 
$Z^c\equiv \{u: u,\theta u\in X^c\}$. Therefore,
$Y^c=\{u+iv: u,v\in Y\}$ and $Z^c=\{u+iv: u,v\in Z\}$, 
where $Y\equiv\{\frac{1+\theta}{2}u: u\in Y^c\}$ and
$Z\equiv\{\frac{1+\theta}{2}u: u\in Z^c\}$, are invariant closed subspaces of
$W$. 
If $Y=\{0\}$ then $Z=\{0\}$ and $Y^c=X^c=\{0\}$, in contradiction
with $X^c$ being non-trivial. 
If $Z=W$ then $Y=W$ and $Z^c=X^c=W^c$, in contradiction 
with $X^c$ being proper.
Therefore $Z=\{0\}$ and $Y=W$, which implies 
$Z^c=\{0\}$ and $Y^c=W^c$.

So, $(M,W)$ is equivalent to $(M,(X^c)^r)$, 
due to the existence of the bijective linear 
map $\alpha:(X^c)^r\to W$, $\alpha(u)=u+\theta u$,
$\alpha^{-1}(u+\theta u)=u$, for $u\in(X^c)^r$.
Suppose that there is a C-conjugation operator of $(M,X^c)$,
$\theta'$. Then $(M,W_{\pm})$ is a proper non-trivial
subsystem of $(M,W)$, where
$W_\pm\equiv\{\frac{1\pm\theta'}{2}w: w\in W\}$, in contradiction with
$(M,W)$ being irreducible.
\end{proof}

\begin{prop}
\label{prop:Rirreducible}
Any real system which is R-real or R-pseudoreal or
R-complex is irreducible.
\end{prop}
\begin{proof}
Consider an irreducible system
on a complex Hilbert space $(M,V)$.
There is a R-imaginary operator $J$ of the system
$(M ,V^{r})$, defined by $J(u)\equiv i u$, for $u\in V^r$.

Let $(M,X^r)$ be a proper non-trivial subsystem of
$(M,V^{r})$. Then $J$ is an R-imaginary operator of
$(M,Y^r)$ and $(M^r,Z^r)$, 
where $Y^r\equiv \{u+J v: u, v \in X^r\}$ and 
$Z^r\equiv \{u: u,Ju\in X^r\}$.
Then $(M,Y)$ and $(M,Z)$ are
subsystems of $(M,V)$,
where the complex Hilbert spaces 
$Y\equiv Y^r$ and $Z\equiv Z^r$ have the scalar 
multiplication such that $(a+ib)(y)=ay+bJy$, 
for $a,b\in\mathbb{R}$ and $y\in Y$ or $y\in Z$.
If $Y=\{0\}$, then $Z=X^r=\{0\}$ which is in contradiction with $X^r$
being non-trivial.
If $Z=V$, then $Y=V$ and $X^r=V^r$ which is in contradiction with
$X^r$ being non-trivial.
So $Z=\{0\}$ and $Y=V$, which implies that $V=(X^r)^c$.

Then there is a C-conjugation operator of $(M,V)$, 
$\theta$, defined by $\theta(u+iv)\equiv u-iv$, for $u,v\in X^r$. 
We have $X^r=V_\theta$. 
Suppose there is a R-imaginary operator of $(M,V_\theta)$, 
$J'$. Then $(M,V_{\pm})$, where
$V_{\pm}\equiv\{\frac{1\pm iJ'}{2}v: v\in V\}$, 
are proper non-trivial subsystems of $(M,V)$, 
in contradiction with $(M,V)$ being irreducible.

Therefore, if $(M,V)$ is C-real, then $(M,V_\theta)$ is
R-real irreducible. 
If $(M,V)$ is C-pseudoreal or C-complex, then 
 $(M,V_\theta^r)$ is R-pseudoreal or R-complex, irreducible. 
\end{proof}

\subsection{Schur Systems}

\begin{defn}[Schur System]
\label{defn:schur}
A system $(M,V)$, on a complex Hilbert space $V$, 
is a Schur system if the set of normal operators of $(M,V)$ is isomorphic to $\mathbb{C}$.\\
Consider an irreducible system $(M,W)$, on a real Hilbert space $W$ and 
let $(M,W^c)$ be its complexification:
1) $(M,W)$ is Schur R-real iff $(M,W^c)$ is Schur C-real;\\
2) $(M,W)$ is Schur R-pseudoreal iff $(M,V)$ is Schur C-pseudoreal,
with $W^c=V\oplus \bar V$;\\
3) $(M,W)$ is Schur R-complex iff $(M,V)$ is Schur C-complex, with
$W^c=V\oplus \bar V$.
\end{defn}

\begin{lem}
Consider a Schur system $(M,V)$ on a complex
Hilbert space. An anti-isometry of $(M,V)$, if it exists, is unique
up to a complex phase.
\end{lem}

\begin{proof}
Let $\theta_1$,$\theta_2$ be two anti-isometries of $(M,V)$. The
product $(\theta_2\theta_1)$ is an isometry of $(M,V)$;
since $(M,V)$ is irreducible,
$(\theta_2\theta_1)=e^{i\phi}$; with $\phi\in \mathbb{R}$.

Therefore $\theta_2=\alpha \theta_1\alpha^{-1}$; where 
$\alpha\equiv e^{i\frac{\phi}{2}}$ is a complex phase.
\end{proof}

\begin{prop}
Two R-real Schur systems are
isometric iff their complexifications are isometric.
\end{prop}
\begin{proof}
Let $(M,V)$ and $(N,W)$ be C-real Schur systems,
with $\theta_M$ and $\theta_N$ the respective C-conjugation operators.
If there is an isometry $\alpha:V\to W$ such that 
$\alpha M=N\alpha$, then 
$\vartheta\equiv \alpha\theta_M\alpha^{-1}$ is an anti-isometry of
$(N,W)$. Since it is unique up to a phase, then
$\theta_N=e^{i\phi}\vartheta$. Therefore $e^{i\frac{\phi}{2}}\alpha$ is an
isometry between $(M,V_\theta)$ and $(N,W_\theta)$, where
$V_{\theta_M}\equiv\{(1+\theta_M)v: v\in V\}$.
\end{proof}

\begin{prop}
Two C-complex or C-pseudoreal Schur systems are isometric or anti-isometric 
iff their realifications are isometric.
\end{prop}
\begin{proof}
Let $(M,V)$ and $(N,W)$ be R-complex or R-pseudoreal Schur
systems, with $J_M$ and $J_N$ the respective R-imaginary
operators. If there is an isometry $\alpha:V\to W$ such that
$\alpha M=N\alpha$, then 
$K\equiv \alpha J_M\alpha^{-1}$ is a R-imaginary operator of
$(N,W)$. 
When considering $(N,W_{J_N})$ and $(M,V_{J_M})$, where 
$W_{J_N}\equiv \{(1-iJ_N) w: w\in W\}$,  we get that 
$(1-J_N K)(1-K J_N)=r$ as an operator of $W_{J_N}$, where $r$ is a
non-negative  null real scalar. If $c=0$ then $K=-J_N$ and $\alpha$
defines an anti-isometry between $(M,V_{J_M})$ and $(N,W_{J_N})$.
If $c\neq 0$ then $(1-J_N K)\alpha c^{-\frac{1}{2}}$ is an isometry
between $(M,V_{J_M})$ and $(N,W_{J_N})$.
\end{proof}

\begin{prop}
The space of normal operators of a R-real Schur system is
isomorphic to $\mathbb{R}$.
\end{prop}
\begin{proof}
Let $(M,V)$ be a C-real Schur system,
with $\theta$ the C-conjugation operator.
If there is an endomorphism $\alpha:V\to V$ such that 
$\alpha M=M\alpha$, we know that $\alpha=re^{i\varphi}$. Then the endomorphism of $V_\theta$
is a real number.
\end{proof}

\begin{prop}
The space of normal operators of a R-complex Schur system is
isomorphic to $\mathbb{C}$.
\end{prop}
\begin{proof}
Let $(M,V)$ be a R-complex Schur
system, with $J$ the R-imaginary
operator. 
If there is a normal operator $\alpha$ of $(M,V)$, then
$KK^\dagger$ is a normal operator of the C-complex Schur system
$(M,V_{J})$, where $K\equiv
(\alpha+J\alpha J)$ and $V_{J}\equiv \{(1-iJ) v: v\in V\}$.
If $KK^\dagger=r>0$, then $\frac{K}{\sqrt{r}}$ is unitary and $V_J$ is
equivalent to $\overline{V}_{J}$ which would imply that $(M,V)$ is
C-pseudoreal. Therefore $K=0$ and hence $\alpha$ is a normal operator of 
$(M,V_{J})$, so $\alpha=r e^{J\theta}$.
\end{proof}

\begin{prop}
The space of normal operators of a R-pseudoreal Schur system is
isomorphic to $\mathbb{H}$ (quaternions).
\end{prop}
\begin{proof}
Let $(M,V)$ be a R-pseudoreal Schur
system, with $J$ the R-imaginary
operator. If there is an endomorphism $\alpha$ of $(M,V)$, then
$SS^\dagger$ and $TT^\dagger$ are a self-adjoint endomorphisms of the C-complex Schur system
$(M,V_{J})$, where $S\equiv(\alpha-J\alpha J)/2$,
$T\equiv(\alpha+J\alpha J)/2$  and 
$V_{J}\equiv\{(1-iJ) v: v\in V\}$.
Let $K$ be an unitary operator of
$(M,V)$ and anti-commuting with $J$, then $K^2=e^{J\theta}$ and 
$Ke^{J\theta}=K(K^2)=(K^2)K=e^{J\theta}K$, therefore $K^2=-1$.
If $TT^\dagger=t>0$, then $\frac{T}{\sqrt{t}}$ is unitary and
anti-commutes with $J$, $TK$ is a normal endomorphism of  
$(M,V_{J})$ and therefore $T=Kc+KJd$; if $TT^\dagger=0$ then 
$c=d=0$.
If $SS^\dagger=s>0$, then $\frac{S}{\sqrt{s}}$ is unitary and
commutes with $J$, $S$ is a normal endomorphism of  
$(M,V_{J})$ and therefore $S=a+Jb$; if $SS^\dagger=0$ then 
$a=b=0$.

Therefore $\alpha=S+T=a+Jb+Kc+KJd$, 
which is isomorphic to the quaternions.
\end{proof}

\subsection{Finite-dimensional representations}
\label{section:Finite}

\begin{lem}[Schur's lemma for finite-dimensional representations\cite{schur}]
\label{lem:commuting}
Consider an irreducible finite-dimensional representation $(M_G,V)$ of
a Lie group $G$ on a complex Hilbert space $V$. If the representation
$(M_G,V)$ is irreducible then any endomorphism $S$ of $(M_G,V)$ is a
complex scalar.
\end{lem}

\begin{lem}
Consider an irreducible complex finite-dimensional representation $(M,V)$ on a
complex Hilbert space. Then there is internal product such that:
1) The system is C-real iff there is an anti-linear involution of $(M,V)$;\\
2) The system is C-pseudoreal iff 
there is not an anti-linear bounded involution of $(M,V)$, but there is an
anti-isomorphism of $(M,V)$;\\
3) The system is C-complex iff 
there is no anti-isomorphism of $(M,V)$.
\end{lem}
\begin{proof}
Let $S$ be an anti-isomorphism of an irreducible representation $(M,V)$.
Then $S^2=re^{i\varphi}$. But $S^2$ commutes with $S$ which is anti-linear, 
so $S^2=\pm r$. So, there is an internal product such that $S$ is anti-unitary.
\end{proof}

\begin{defn}
A finite-dimensional system is completely reducible iff it can
be expressed as a direct sum of irreducible systems.
\end{defn}

\begin{rmk}[Weyl theorem]
All finite-dimensional representations of a semi-simple Lie group 
(such as SL(2,C)) are completely reducible.
\end{rmk}

\subsection{Unitary representations and Systems of Imprimitivity}
\label{section:Unitary}
\begin{defn}[Normal System]
A System $(M,V)$ is normal iff $M$ is a set $M$ of normal operators on $V$
closed under Hermitian conjugation---for all $m\in M$ there is $n\in M$ such that
$n=m^\dagger$.
\end{defn}

A unitary representation or a System of Imprimitivity are examples of a normal System.

\begin{rmk}
$W^\bot$ is the orthogonal complement of the subspace $W$ of the
Hilbert space $V$ if:\\
1) $V=W \oplus W^\bot$, that is, 
all $v\in V$ can be expressed as $v=w+x$, where $w\in W$ and $x\in W^\bot$;\\
2) if $w\in W$ and $x\in W^\bot$, then $x^\dagger w=0$.
\end{rmk}

\begin{lem}
\label{lem:orthogonal}
Consider a normal system $(M,V)$. Then, for all subsystem
$(M,W)$ of $(M_G,V)$, $(M_G,W^\bot)$ is also a subsystem of
$(M,V)$, where $W^\bot$ is the orthogonal complement of the subspace $W$.
\end{lem}

\begin{proof} 
Let $(M,W)$ be a subsystem of $(M,V)$.
$W^\bot$ is the orthogonal complement of $W$.

For all $x\in W^\bot$, $w\in W$ and $m\in M$, 
$<m x, w>=<x,m^\dagger w>$. 

Since $W$ is invariant and there is $n\in M$, such
that $n=m^\dagger$, then $w'\equiv (m^\dagger w)\in W$.
 
Since $x\in W^\bot$ and $w'\in W$, then $<x, w'>=0$.

This implies that if $x\in W^\bot$), also $(m x)\in W^\bot$, for all $m\in M$.
\end{proof}

\begin{lem}
\label{lem:commuting}
Any Schur normal system on a complex Hilbert space is irreducible.
\end{lem}

\begin{proof}
Let $(M,W)$ and $(M,W^\bot)$ be sub-systems of the complex Schur system $(M,V)$,
where $W^\bot$ is the orthogonal complement of $W$.

There is a bounded endomorphism $P: V\to V$, such that, 
for $w,w'\in W$, $x,x'\in W^\bot$, $P(w+x)=w$. $P^2=P$ and $P$ is
hermitian:
\begin{align}
&<w'+x',P (w+x)>=<w',w>=
<P(w'+x'),w+x>
\end{align}
Let $w'\equiv m w\in W$ and $x'\equiv m x\in
W^\bot$:
\begin{align}
m P(w+x)&=m w=w'\\
P m(w+x)&=P(w'+x')=w'
\end{align}
Which implies that $P$ commutes with all $m\in M$, so  $P\in\{0,1\}$. 
If $P=1$, then $W=V$, if $P=0$, then $W$ is the null space.
\end{proof}

So a complex Schur normal system is irreducible, and hence, from Defns.\ref{defn:Rsystem},\ref{defn:schur} and Prop.\ref{prop:Rirreducible}, 
a real Schur normal system is also irreducible.

\begin{lem}[Schur's lemma for unitary representations\cite{schur}]
\label{lem:commuting}
Consider an irreducible unitary representation $(M,V)$ of a Lie
group $G$ on a complex Hilbert space $V$. If the representation
$(M,V)$ is irreducible then any normal operator $N$ of $(M,V)$ is a
scalar.
\end{lem}

\begin{defn}
A unitary system is completely reducible iff it can be expressed as a
direct integral of irreducible systems.
\end{defn}

\begin{rmk}
All unitary representations of a separable locally compact group 
(such as the Poincare group) are completely reducible.
\end{rmk}

\begin{defn}
Consider a measurable space $(X, M)$, where $M$ is a $\sigma$-algebra
of subsets of $X$. A projection-valued-measure, $\pi$, is a map from
$M$ to the set of self-adjoint projections on a Hilbert space $H$ such
that $\pi(X)$ is the identity operator on $H$ and the function
$<\psi,\pi(A)\psi>$, with $A\in M$ is a measure on $M$, for all
$\psi\in H$.
\end{defn}

\begin{defn}
Suppose now that $X$ is a representation of $G$. 
Then, a system of imprimitivity is a pair $(U,\pi)$, where $\pi$ is a
projection valued measure and $U$ an unitary representation of $G$ on
the Hilbert space $H$, such that $U(g)\pi(A) U^{-1}(g)=\pi(gA)$.
\end{defn}

\begin{rmk}[Imprimitivity Theorem (thrm 6.12 \cite{commutingring,mackey,vara,squareimprimitivity})]
\label{lem:commuting} Let G be a Lie group, H its closed subgroup.
Let a pair (V, E) be a system of imprimitivity for G based on G/H on a separable complex Hilbert space.
Then there exists a representation L of H such that (V, E) is equivalent to
the canonical system of imprimitivity (V L , E L ). For any two representations L, L' of the subgroup 
H the corresponding canonical systems of imprimitivity are equivalent if and only
if L, L' are equivalent. The sets of normal operators commuting of C(VL, EL ) and of C(L) are isomorphic.
\end{rmk}

So we can define a map from the real to the complex systems of imprimitivity---analogous to the 
one for unitary representations.

\section{Finite-dimensional representations of the Lorentz group}
\label{section:Lorentz}
\subsection{Majorana, Dirac and Pauli Matrices and Spinors}
\begin{defn}
$\mathbb{F}^{m\times n}$ is the vector space of $m\times n$ matrices whose
entries are elements of the field $\mathbb{F}$.
\end{defn}

In the next remark we state the Pauli's fundamental theorem of gamma
matrices. The proof can be found in the reference\cite{diracmatrices}.
\begin{rmk}[Pauli's fundamental theorem]
\label{rem:fundamental}
Let $A^\mu$, $B^\mu$, $\mu\in\{0,1,2,3\}$, be two sets of
$4\times 4$ complex matrices verifying:
\begin{align}
A^\mu A^\nu+A^\nu A^\mu&=-2\eta^{\mu\nu}\\
B^\mu B^\nu+B^\nu B^\mu&=-2\eta^{\mu\nu}
\end{align}
Where $\eta^{\mu\nu}\equiv diag(+1,-1,-1-1)$ is the Minkowski metric.

1) There is an invertible complex matrix $S$ such that
$B^\mu=S A^\mu S^{-1}$, for all $\mu\in\{0,1,2,3\}$. 
$S$ is unique up to a non-null scalar.

2) If $A^\mu$ and $B^\mu$ are all unitary, then $S$ is unitary.
\end{rmk}

\begin{prop}
\label{prop:realsimilar}
Let $\alpha^\mu$, $\beta^\mu$, $\mu\in\{0,1,2,3\}$, be two sets of
$4\times 4$ real matrices verifying:
\begin{align}
\alpha^\mu\alpha^\nu+\alpha^\nu\alpha^\mu&=-2\eta^{\mu\nu}\\
\beta^\mu\beta^\nu+\beta^\nu\beta^\mu&=-2\eta^{\mu\nu}
\end{align}
Then there is a real matrix $S$, with $|det S|=1$, such that
$\beta^\mu=S\alpha^\mu S^{-1}$, for all  $\mu\in\{0,1,2,3\}$. $S$ is unique up to a signal. 
\end{prop}

\begin{proof}
From remark \ref{rem:fundamental}, we know that there is an
invertible matrix $T'$, unique up to a non-null scalar, such that $\beta^\mu=T'\alpha^\mu
T^{'-1}$.
Then $T\equiv T'/|det(T')|$ has $|det T|=1$ and it is unique up to a
complex phase.

Conjugating the previous equation, we get $\beta^\mu=T^*\alpha^\mu
T^{*-1}$.
Then $T^*=e^{i 2 \theta} T$ for some real number $\theta$.
Therefore $S\equiv e^{i \theta}T$ is a real matrix,
with $|det S|=1$, unique up to a signal.
\end{proof}

\begin{defn}
The Majorana matrices, $i\gamma^\mu$, $\mu\in\{0,1,2,3\}$, are $4\times
4$ complex unitary matrices verifying:
\begin{align}
(i\gamma^\mu)(i\gamma^\nu)+(i\gamma^\nu)(i\gamma^\mu)&=-2\eta^{\mu\nu}
\end{align}
The Dirac matrices are $\gamma^\mu\equiv
-i(i\gamma^\mu)$.
\end{defn}

In the Majorana bases, the Majorana matrices are $4\times 4$ real
orthogonal matrices. An example of the Majorana matrices in a
particular Majorana basis is:
\begin{align}
\begin{array}{llllll}
\label{basis}
i\gamma^1=&\left[ \begin{smallmatrix}
+1 & 0 & 0 & 0 \\
0 & -1 & 0 & 0 \\
0 & 0 & -1 & 0 \\
0 & 0 & 0 & +1 \end{smallmatrix} \right]&
i\gamma^2=&\left[ \begin{smallmatrix}
0 & 0 & +1 & 0 \\
0 & 0 & 0 & +1 \\
+1 & 0 & 0 & 0 \\
0 & +1 & 0 & 0 \end{smallmatrix} \right]&
i\gamma^3=\left[ \begin{smallmatrix}
0 & +1 & 0 & 0 \\
+1 & 0 & 0 & 0 \\
0 & 0 & 0 & -1 \\
0 & 0 & -1 & 0 \end{smallmatrix} \right]\\
\\
i\gamma^0=&\left[ \begin{smallmatrix}
0 & 0 & +1 & 0 \\
0 & 0 & 0 & +1 \\
-1 & 0 & 0 & 0 \\
0 & -1 & 0 & 0 \end{smallmatrix} \right]&
i\gamma^5=&\left[ \begin{smallmatrix}
0 & -1 & 0 & 0 \\
+1 & 0 & 0 & 0 \\
0 & 0 & 0 & +1 \\
0 & 0 & -1 & 0 \end{smallmatrix} \right]&
=-\gamma^0\gamma^1\gamma^2\gamma^3
\end{array}
\end{align}

In reference \cite{realgamma} it is proved that the set of five
anti-commuting $4\times 4$ real matrices is unique up to
isomorphisms. So, for instance, with $4\times 4$ real matrices it is not possible
to obtain the euclidean signature for the metric.

\begin{defn}
The Dirac spinor is a $4\times 1$ complex column matrix, $\mathbb{C}^{4\times 1}$.
\end{defn}

The space of Dirac spinors is a 4 dimensional complex vector space.

\begin{lem}
The charge conjugation operator $\Theta$, is an anti-linear involution
commuting with the Majorana matrices $i\gamma^\mu$. 
It is unique up to a complex phase.
\end{lem}

\begin{proof}
In the Majorana bases, the complex conjugation is a charge conjugation
operator. Let $\Theta$ and $\Theta'$ be two charge conjugation operators
operators. Then, $\Theta\Theta'$ is a complex invertible matrix commuting with
$i\gamma^\mu$, therefore, from Pauli's fundamental theorem,  
$\Theta\Theta'=c$, where $c$ is a non-null complex scalar.
Therefore $\Theta'=c^*\Theta$ and from $\Theta'\Theta'=1$, we get that
$c^* c=1$.
\end{proof}

\begin{defn}
Let $\Theta$ be a charge conjugation operator.

The set of Majorana spinors, $Pinor$, is the set of Dirac spinors
verifying the Majorana condition (defined up to a complex phase):
\begin{align}
Pinor\equiv \{u\in \mathbb{C}^{4\times 1}: \Theta u= u\}
\end{align}
\end{defn}

The set of Majorana spinors is a 4 dimensional real vector space. 
Note that the linear combinations of
Majorana spinors with complex scalars do not verify the Majorana
condition.

There are 16 linear independent products of Majorana matrices. These
form a basis of the real vector space of endomorphisms of Majorana spinors,
$End(Pinor)$. In the Majorana bases, $End(Pinor)$ is the vector space of
$4\times 4$ real matrices.

\begin{defn}
The Pauli matrices $\sigma^k,\ k\in\{1,2,3\}$ are $2\times 2$
hermitian, unitary, anti-commuting, complex matrices.
The Pauli spinor is a  $2\times 1$ complex column matrix. The space of
Pauli spinors is denoted by $Pauli$.
\end{defn}

The space of Pauli spinors, $Pauli$, is a 2 dimensional complex vector
space and a 4 dimensional real vector space. The realification of
the space of Pauli spinors is isomorphic to the space of Majorana
spinors.

\subsection{On the Lorentz, SL(2,C) and Pin(3,1) groups}

\begin{rmk} 
The Lorentz group, $O(1,3)\equiv\{\lambda \in \mathbb{R}^{4\times 4}: \lambda^T \eta \lambda=\eta \}$, is the set of
real matrices that leave the metric, $\eta=diag(1,-1,-1,-1)$,
invariant.

The proper orthochronous Lorentz subgroup is defined by
$SO^+(1,3)\equiv\{\lambda \in
O(1,3): det(\lambda)=1, \lambda^0_{\ 0}>0 \}$. 
It is a normal subgroup. 
The discrete Lorentz subgroup of parity and time-reversal is 
$\Delta \equiv \{1,\eta,-\eta,-1\}$.

The Lorentz group is the semi-direct product of the previous
subgroups, $O(1,3)=\Delta \ltimes SO^+(1,3)$.  
\end{rmk}

\begin{defn}
The set $Maj$ is the 4 dimensional real space of the linear
combinations of the Majorana matrices, $i\gamma^\mu$:
\begin{align}
Maj\equiv\{a_\mu i\gamma^\mu: a_\mu\in \mathbb{R},\ \mu\in\{0,1,2,3\}\}
\end{align}
\end{defn}

\begin{defn}
$Pin(3,1)$ \cite{pin} is the group of endomorphisms of Majorana
spinors that leave the space $Maj$ invariant, that is:
\begin{align}
Pin(3,1)\equiv 
\Big\{S\in End(Pinor):\ |det S|=1,\ S^{-1}(i\gamma^\mu)S\in Maj,\ \mu\in\{0,1,2,3\} \Big\}
\end{align}
\end{defn}

\begin{prop}
\label{prop:map}
The map $\Lambda:Pin(3,1)\to O(1,3)$ defined by:
\begin{align}
(\Lambda(S))^\mu_{\ \nu}i\gamma^\nu\equiv S^{-1}(i\gamma^\mu)S
\end{align}
is two-to-one and surjective. It defines a group homomorphism.
\end{prop}

\begin{proof}
1) Let $S\in Pin(3,1)$. Since the Majorana matrices are a basis of the
real vector space $Maj$, there is an unique real matrix $\Lambda(S)$ such that:
\begin{align}
(\Lambda(S))^\mu_{\ \nu}i\gamma^\nu=S^{-1}(i\gamma^\mu)S
\end{align}
Therefore, $\Lambda$ is a map with domain $Pin(3,1)$. Now we can check
that $\Lambda(S)\in O(1,3)$:
\begin{align}
&(\Lambda(S))^\mu_{\ \alpha}\eta^{\alpha\beta}(\Lambda(S))^\nu_{\
  \beta}=-\frac{1}{2}(\Lambda(S))^\mu_{\
  \alpha}\{i\gamma^\alpha,i\gamma^\beta\}(\Lambda(S))^\nu_{\
  \beta}=\\
&=-\frac{1}{2}S\{i\gamma^\mu,i\gamma^\nu\}S^{-1}=S\eta^{\mu\nu}S^{-1}=\eta^{\mu\nu}
\end{align}
We have proved that $\Lambda$ is a map from $Pin(3,1)$ to $O(1,3)$.

2) Since any $\lambda\in O(1,3)$ conserve the metric $\eta$, the matrices
$\alpha^\mu\equiv \lambda^\mu_{\ \nu} i\gamma^\nu$ verify:
\begin{align}
\{\alpha^\mu,\alpha^\nu\}=-2\lambda^\mu_{\ \alpha}\eta^{\alpha\beta}\lambda^\nu_{\ \beta}=-2\eta^{\mu\nu}
\end{align}
In a basis where the Majorana matrices are real, from Proposition
\ref{prop:realsimilar} there is a real invertible matrix $S_\lambda$,
with $|det S_\Lambda|=1$, such that $\lambda^\mu_{\ \nu} i\gamma^\nu=S^{-1}_\lambda
(i\gamma^\mu)S_\lambda$. 
The matrix $S_\Lambda$ is unique up to a sign. So, $\pm S_\lambda\in
Pin(3,1)$ and we proved that the map
 $\Lambda:Pin(3,1)\to O(1,3)$ is two-to-one and surjective.

3) The map defines a group homomorphism because:
\begin{align}
&\Lambda^\mu_{\ \nu}(S_1)\Lambda^\nu_{\
  \rho}(S_2)i\gamma^\rho=\Lambda^\mu_{\ \nu}S_2^{-1}i\gamma^\nu S_2\\
&=S_2^{-1}S_1^{-1}i\gamma^\mu S_1 S_2=\Lambda^\mu_{\ \rho}(S_1 S_2)i\gamma^\rho
\end{align}
\end{proof}

\begin{rmk}
\label{rem:SL(2,C)}
The group $SL(2,\mathbb{C})=\{e^{\theta^j
  i\sigma^j+b^j\sigma^j}: \theta^j,b^j\in
\mathbb{R},\ j\in\{1,2,3\}\}$ is simply connected. 
Its projective representations are equivalent to its ordinary representations\cite{weinberg}.

There is a two-to-one, surjective map $\Upsilon:SL(2,\mathbb{C})\to
SO^+(1,3)$, defined by:
\begin{align}
\Upsilon^{\mu}_{\ \nu}(T)\sigma^\nu\equiv T^\dagger \sigma^\mu T
\end{align}
Where $T\in SL(2,\mathbb{C})$, $\sigma^0=1$ and $\sigma^j$, $j\in\{1,2,3\}$ are the Pauli matrices.
\end{rmk}

\begin{lem}
Consider that $\{M_+,M_-,i\gamma^5 M_+,i\gamma^5 M_-\}$ and $\{P_+,P_-,iP_+,iP_-\}$ are orthonormal basis of
the 4 dimensional real vector spaces $Pinor$ and $Pauli$, respectively, verifying:
\begin{align}
\gamma^0\gamma^3 M_\pm=\pm M_\pm&,\ \sigma^3 P_\pm=\pm P_\pm
\end{align}
The isomorphism
$\Sigma:Pauli \to Pinor$ is defined by:
\begin{align}
\Sigma(P_+)=M_+,&\ \Sigma(iP_+)=i\gamma^5 M_+\\
\Sigma(P_-)=M_-,&\ \Sigma(iP_-)=i\gamma^5 M_-
\end{align}

The group $Spin^+(3,1)\equiv \{\Sigma\circ A \circ \Sigma^{-1}:A\in
SL(2,\mathbb{C})\}$ is a subgroup of $Pin(1,3)$. 
For all $S\in Spin^+(1,3)$, $\Lambda(S)=\Upsilon(\Sigma^{-1}\circ S \circ \Sigma)$.
\end{lem}

\begin{proof}
From remark \ref{rem:SL(2,C)},  $Spin^+(3,1)=\{e^{\theta^j
  i\gamma^5\gamma^0\gamma^j+b^j\gamma^0\gamma^j}: \theta^j,b^j\in
\mathbb{R},\ j\in\{1,2,3\}\}$.
Then, for all $T\in SL(2,C)$:
\begin{align}
-i\gamma^0 \Sigma \circ T^\dagger
\circ \Sigma^{-1} i\gamma^0&=\Sigma \circ T^{-1}
\circ \Sigma^{-1}
\end{align}
Now, the map $\Upsilon:SL(2,\mathbb{C})\to
SO^+(1,3)$ is given by:
\begin{align}
\Upsilon^{\mu}_{\ \nu}(T)i\gamma^\nu = (\Sigma \circ T^{-1}
\circ \Sigma^{-1}) i\gamma^\mu (\Sigma\circ T \circ \Sigma^{-1})
\end{align}
Then, all $S\in Spin^+(3,1)$ leaves the space $Maj$ invariant:
\begin{align}
S^{-1} i\gamma^\mu S=
\Upsilon^{\mu}_{\ \nu}(\Sigma^{-1}\circ S \circ \Sigma)i\gamma^\nu
\in Maj
\end{align}
Since all the products of Majorana matrices, except the identity, are
traceless, then $det(S)=1$. So,
$Spin^+(3,1)$ is a subgroup of $Pin(1,3)$ and $\Lambda(S)=\Upsilon(\Sigma^{-1}\circ S \circ \Sigma)$.
\end{proof}

\begin{defn}
The discrete Pin subgroup $\Omega\subset Pin(3,1)$ is:
\begin{align}
\Omega \equiv \{\pm 1,\pm i\gamma^0,\pm \gamma^0\gamma^5,\pm
i\gamma^5\}
\end{align}
\end{defn}

The previous lemma and the fact that $\Lambda$ is continuous, 
implies that $Spin^+(1,3)$ is a double cover of $SO^+(3,1)$.
We can check that for all $\omega\in \Omega$, $\Lambda(\pm \omega)\in
\Delta$. 
That is, the discrete Pin subgroup is the double cover of the
discrete Lorentz subgroup. Therefore, $Pin(3,1)=\Omega \ltimes Spin^+(1,3)$

Since there is a two-to-one continuous surjective group homomorphism,
$Pin(3,1)$ is a double cover of $O(1,3)$, $Spin^+(3,1)$ 
is a double cover of $SO^+(1,3)$ and $Spin^+(1,3)\cap SU(4)$ is a
double cover of $SO(3)$. We can check that $Spin^+(1,3)\cap SU(4)$ is
equivalent to $SU(2)$.

\subsection{Finite-dimensional representations of SL(2,C)}

\begin{rmk}
Since SL(2,C) is a semisimple Lie group, all its finite-dimensional 
(real or complex) representations are direct sums of irreducible
representations.
\end{rmk}

\begin{rmk} 
The finite-dimensional complex irreducible representations of SL(2,C) 
are labeled by $(m,n)$, where $2m,2n$ are natural numbers. 
Up to equivalence, the representation space $V_{(m,n)}$ is the tensor
product of the complex vector spaces $V_m^+$ and $V_n^-$, where $V_m^\pm$ is a
symmetric tensor with $2m$ Dirac spinor indexes, such that 
$\gamma^5_{\  k}v=\pm v$, where $v\in V_m^\pm$ and $\gamma^5_{\  k}$
is the Dirac matrix $\gamma^5$ acting on the $k$-th index of $v$.

The group homomorphism consists in applying the same matrix of
$Spin^+(1,3)$, correspondent to the $SL(2,C)$ group element we are
representing, to each index of $v$. 
$V_{(0,0)}$ is equivalent to $\mathbb{C}$ and the image of the group
homomorphism is the identity.

These are also projective representations of the time reversal transformation,
but, for $m\neq n$, not of the parity transformation, that is, 
under the parity transformation, $(V^+_m\otimes V^-_n)\to (V^-_m\otimes V^+_n)$ and under the time
reversal transformation $(V^+_m\otimes V^-_n)\to (V^+_m\otimes V^-_n)$.
\end{rmk}

\begin{lem}
The finite-dimensional real irreducible representations of SL(2,C) 
are labeled by $(m,n)$, where $2m,2n$ are natural numbers and $m\geq
n$.
Up to equivalence, the representation space $W_{(m,n)}$ is defined
for $m\neq n$ as:
\begin{align*}
W_{(m,n)}&\equiv \{\frac{1+(i\gamma^5)_1\otimes
   (i\gamma^5)_1}{2}w: w\in W_m\otimes W_n\}\\
W_{(m,m)}&\equiv 
\{\frac{1+(i\gamma^5)_1\otimes(i\gamma^5)_1}{2}w: w\in (W_m)^2\}
\end{align*} 
where $W_m$ is a
symmetric tensor with $m$ Majorana spinor indexes, such that 
$(i\gamma^5)_{1}(i\gamma^5)_{k}w=-w$, where $w\in W_m$; $(i\gamma^5)_{k}$
is the Majorana matrix $i\gamma^5$ acting on the $k$-th index of $w$;
$(W_m)^2$ is the space of the linear combinations of the
symmetrized tensor products $(u\otimes v+v\otimes u)$, for $u,v\in W_m$.  

The group homomorphism consists in applying the same matrix of
$Spin^+(1,3)$, correspondent to the $SL(2,C)$ group element we are
representing, to each index of the tensor. In the $(0,0)$ case,
$W_{(0,0)}$ is equivalent to $\mathbb{R}$ and the
image of the group homomorphism is the identity.

These are also projective representations of the full Lorentz group, 
that is, under the parity or time reversal transformations,
$(W_{m,n}\to W_{m,n})$.
\end{lem}

\begin{proof}
For $m\neq n$ the complex irreducible representations of SL(2,C) are
C-complex. The complexification of $W_{(m,n)}$ verifies
$W_{(m,n)}^c=(V^+_m\otimes V^-_n)\oplus (V^-_m\otimes V^+_n)$.

For $m=n$ the complex irreducible representations of SL(2,C) are
C-real. In a Majorana basis, the C-conjugation operator of
$V_{(m,m)}$, $\theta$, is defined as 
$\theta(u\otimes v)\equiv v^*\otimes u^*$, where $u\in V^+_m$ and $v\in V^-_m$. 
We can check that there is a bijection $\alpha:W_{(m,m)}\to
(V_{(m,m)})_\theta$, 
defined by $\alpha(w)\equiv \frac{1-i(i\gamma^5)_1\otimes 1}{2}w$; 
$\alpha^{-1}(v)\equiv v+v^*$, for $w\in W_{(m,m)}$, $v\in (V_{(m,m)})_\theta$.

Using the map from Section 2,
we can check that the representations 
$W_{(m,n)}$, with $m\geq n$, are the unique 
finite-dimensional real irreducible representations of SL(2,C), up to
isomorphisms.

We can check that $W_{(m,n)}^c$ is equivalent to $W_{(n,m)}^c$,
therefore, invariant under the parity or time reversal transformations.
\end{proof}

As examples of real irreducible representations of $SL(2,C)$ we have
for $(1/2,0)$ the Majorana spinor, for $(1/2,1/2)$ the linear
combinations of the matrices $\{1,\gamma^0\vec{\gamma}\}$, for $(1,0)$ the linear 
combinations of the matrices $\{i\vec{\gamma},\vec{\gamma}\gamma^5\}$. The group 
homomorphism is defined as $M(S)(u)\equiv Su$ and $M(S)(A)\equiv S A S^\dagger$, 
for $S\in Spin^+(1,3)$,
$u\in Pinor$, $A\in \{1,\vec{\gamma}\gamma^0\}$ or $A\in
\{i\vec{\gamma},\vec{\gamma}\gamma^5\}$. 

We can check that the domain
of $M$ can be extended to $Pin(1,3)$, leaving the considered vector
spaces invariant. For $m=n$, we can define the ``pseudo-representation'' 
$W_{(m,m)}'\equiv \{((i\gamma^5)_1\otimes 1) w: w\in W_{(m,m)}\}$
which is equivalent to $W_{(m,m)}$ as an $SL(2,C)$ representation, but
under parity transforms with the opposite sign.
As an example, the ``pseudo-representation'' $(1/2,1/2)$ is defined as
the linear combinations of the matrices $\{i\gamma^5,i\gamma^5\vec{\gamma}\gamma^0\}$.

\section{Unitary representations of the Poincare group}
\label{section:Poincare}

\subsection{Bargmann-Wigner fields}
\begin{defn}
\label{defn:Theta}
Consider that $\{M_+,M_-,i\gamma^0M_+,i\gamma^0M_-\}$ and $\{P_+,P_-,iP_+,iP_-\}$ are orthonormal basis 
of the 4 dimensional real vector spaces $Pinor$ and $Pauli$, respectively, verifying:
\begin{align*}
\gamma^3\gamma^5 M_\pm=\pm M_\pm&,\ \sigma^3 P_\pm=\pm P_\pm
\end{align*}
Let $H$ be a real Hilbert space. 
For all $h\in H$, the bijective linear map
$\Theta_H:Pauli\otimes_{\mathbb{R}} H\to Pinor\otimes_{\mathbb{R}}H$ is defined by:
\begin{align*}
\Theta_H(h\otimes_{\mathbb{R}} P_+)=h\otimes_{\mathbb{R}} M_+,&\ \Theta_H(h \otimes_{\mathbb{R}}
iP_+)=h\otimes_{\mathbb{R}} i\gamma^0 M_+\\
\Theta_H(h\otimes_{\mathbb{R}} P_-)=h\otimes_{\mathbb{R}} M_-,&\ \Theta_H(h\otimes_{\mathbb{R}} iP_-)=h\otimes_{\mathbb{R}}i\gamma^0 M_-
\end{align*}
\end{defn}

\begin{defn}
Let $H_n$, with $n\in\{1,2\}$, be two real Hilbert spaces 
and $U:Pauli\otimes_{\mathbb{R}} H_1\to
Pauli\otimes_{\mathbb{R}} H_2$ be an operator.
The operator $U^\Theta:Pinor\otimes_{\mathbb{R}} H_1\to
Pinor\otimes_{\mathbb{R}} H_2$ is defined as
$U^\Theta\equiv
\Theta_{H_2}\circ U\circ \Theta^{-1}_{H_1}$.
\end{defn}

The space of Majorana spinors is isomorphic to
the realification of the space of Pauli spinors.

\begin{defn}
The real Hilbert space 
$Pinor(\mathbb{X})\equiv Pinor\otimes L^2(\mathbb{X})$ 
is the space of square integrable
functions with domain $\mathbb{X}$ and image in $Pinor$.
\end{defn}

\begin{defn}
The complex Hilbert space 
$Pauli(\mathbb{X})\equiv Pauli\otimes L^2(\mathbb{X})$ 
is the space of square integrable functions with domain $\mathbb{X}$
and image in $Pauli$.
\end{defn}

\begin{rmk}
The Fourier Transform 
$\mathcal{F}_P: Pauli(\mathbb{R}^3)\to Pauli(\mathbb{R}^3)$ is an unitary operator defined by:
\begin{align*}
\mathcal{F}_P\{\psi\}(\vec{p})\equiv\int d^n\vec{x} \frac{e^{-i\vec{p}\cdot
  \vec{x}}}{\sqrt{(2\pi)^n}}\psi(\vec{x}),\ \psi\in Pauli(\mathbb{R}^3)
\end{align*}
Where the domain of the integral is $\mathbb{R}^3$.
\end{rmk}

\begin{rmk}
The inverse Fourier transform verifies:
\begin{align*}
-\vec{\partial}^2\
\mathcal{F}_P^{-1}\{\psi\}(\vec{x})&=
(\mathcal{F}_P^{-1}\circ R)\{\psi\}(\vec{x})\\
i\vec{\partial}_k
\ \mathcal{F}_P^{-1}\{\psi\}(\vec{x})&=
(\mathcal{F}_P^{-1}\circ R_k')\{\psi\}(\vec{x})
\end{align*}
Where $\psi\in Pauli(\mathbb{R}^3)$ and
$R,R_k':Pauli(\mathbb{R}^3)\to Pauli(\mathbb{R}^3)$, with
$k\in\{1,2,3\}$, are linear maps defined by:
\begin{align*}
R\{\psi\}(\vec{p})&\equiv (\vec{p})^2 \psi(\vec{p})\\
R_k'\{\psi\}(\vec{p})&\equiv \vec{p}_k\ \psi(\vec{p})
\end{align*}
\end{rmk}

\begin{defn}
Let $\vec{x}\in \mathbb{R}^3$. The spherical coordinates
parametrization is:
\begin{align*}
\vec{x}=r(\sin(\theta)\sin(\varphi)\vec{e_1}+\sin(\theta)\sin(\varphi)\vec{e_2}+\cos(\theta)\vec{e}_3)
\end{align*}
where $\{\vec{e}_1,\vec{e}_2,\vec{e}_3\}$ is a fixed orthonormal basis of
$\mathbb{R}^3$ and $r\in [0,+\infty[$, $\theta \in [0,\pi]$, $\varphi
\in [-\pi,\pi]$.
\end{defn}

\begin{defn}
Let
\begin{align*}
\mathbb{S}^3\equiv \{(p,l,\mu):p\in \mathbb{R}_{\geq 0}; 
l,\mu \in \mathbb{Z}; l\geq 0; -l \leq \mu\leq l\}
\end{align*}
The Hilbert space $L^2(\mathbb{S}^3)$ is the real Hilbert space of real
Lebesgue square integrable functions of $\mathbb{S}^3$. The internal product is:
\begin{align*}
<f,g>=\sum_{l=0}^{+\infty}\sum_{\mu=-l}^{l-1}\int_0^{+\infty} dp f(p,l,\mu)
g(p,l,\mu),\ f,g\in L^2(\mathbb{S}^3)
\end{align*}
\end{defn}

\begin{defn}
The Spherical transform $\mathcal{H}_{P}: Pauli(\mathbb{R}^3)\to Pauli(\mathbb{S}^3)$ 
is an operator defined by:
\begin{align*}
\mathcal{H}_{P}\{\psi\}(p,l,\mu)\equiv\int r^2 dr d(\cos\theta)d\varphi
\frac{2 p}{\sqrt{2\pi}}j_l(pr)Y_{l\mu}(\theta,\varphi)\psi(r,\theta,\varphi),\ \psi\in Pauli(\mathbb{R}^3)
\end{align*}
The domain of the integral is $\mathbb{R}^3$. The spherical
Bessel function of the first kind $j_l$ \cite{bessel},
the spherical harmonics $Y_{l\mu}$\cite{harmonics} and the associated Legendre
functions of the first kind $P_{l\mu}$ are:
\begin{align*}
j_l(r)\equiv& r^l\Big(-\frac{1}{r}\frac{d}{dr}\Big)^l \frac{\sin
  r}{r}\\
Y_{l\mu}(\theta,\varphi)\equiv&\sqrt{\frac{2l+1}{4\pi}\frac{(l-m)!}{(l+m)!}}
P_{l}^\mu(\cos\theta)e^{i\mu \varphi}\\
P_{l}^\mu(\xi)\equiv&\frac{(-1)^{\mu}}{2^{l}l!}(1-\xi^{2})^{\mu/2}
\frac{\mathrm{d}^{l+\mu}}{\mathrm{d}\xi^{l+\mu}}(\xi^{2}-1)^{l}
\end{align*}
\end{defn}

\begin{rmk}
Due to the properties of spherical harmonics and Bessel functions, the
Spherical transform  is an unitary operator. The inverse Spherical
transform verifies:
\begin{align*}
-\vec{\partial}^2\
\mathcal{H}_P^{-1}\{\psi\}(\vec{x})&=
(\mathcal{H}_P^{-1}\circ R)\{\psi\}(\vec{x})\\
(-x^1i\partial_2+x^2i\partial_1)
\ \mathcal{H}_P^{-1}\{\psi\}(\vec{x})&=
(\mathcal{H}_P^{-1}\circ R')\{\psi\}(\vec{x})
\end{align*}
Where $\psi\in Pauli(\mathbb{S}^3)$ and
$R,R':Pauli(\mathbb{S}^3)\to Pauli(\mathbb{S}^3)$ 
are linear maps defined by:
\begin{align*}
R\{\psi\}(p,l,\mu)&\equiv p^2 \psi(p,l,\mu)\\
R'\{\psi\}(p,l,\mu)&\equiv \mu\ \psi(p,l,\mu)
\end{align*}
\end{rmk}

\begin{defn}
The real vector space $Pinor_j$, with $2j$ a positive integer, is the
space of linear combinations of the tensor products of $2j$
Majorana spinors, symmetric on the spinor indexes. The real vector
space $Pinor_0$ is the space of linear combinations of the tensor
products of $2$ Majorana spinors, anti-symmetric on the spinor
indexes.
\end{defn}

\begin{defn}
The real Hilbert space 
$Pinor_j(\mathbb{X})\equiv Pinor_j\otimes L^2(\mathbb{X})$ is the
space of square integrable functions with domain $\mathbb{X}$ and
image in $Pinor_j$.
\end{defn}

\begin{defn}
The Hilbert space $Pinor_{j,n}$, with $(j-\nu)$ an
integer and $-j\leq n \leq j$ is defined as:
\begin{align*}
Pinor_{j,n}\equiv \{\Psi\in Pinor_{j}:
\sum_{k=1}^{k=2j}(\gamma^0)_1\Big(\gamma^0\gamma^3\gamma^5\Big)_k\Psi
=2n \Psi\}
\end{align*}
Where $\Big(\gamma^3\gamma^5\Big)_k$ is the matrix 
$\gamma^3\gamma^5$ acting on the Majorana index $k$. 
\end{defn}

\begin{defn}
The Spherical transform 
$\mathcal{H}_{P}': Pinor_j(\mathbb{R}^3)\to Pinor_j(\mathbb{S}^3)$ 
is an operator defined by:
\begin{align*}
\mathcal{H}_{P}'\{\psi\}(p,l,J,\nu)\equiv
\sum_{\mu=-l}^l\sum_{n=-j}^j
<l\mu jn|J\nu>\Big(\mathcal{H}_{P}^{\Theta}\Big)_1\{\psi\}(p,l,\mu,n),
\ \psi\in Pinor_j(\mathbb{R}^3)
\end{align*}
$<l\mu jn|J\nu>$ are the Clebsh-Gordon coefficients and
$\psi(p,l,\mu,n)\in Pinor_{j,n}$ such that
$\psi(p,l,\mu)=\sum_{n=-j}^j \psi(p,l,\mu,n)$. $(j-n)$, $(J-\nu)$ and
$(J-j)$ are integers, with $-J\leq\nu\leq J$ and $|j-l|\leq J \leq
j+l$.
$\Big(\mathcal{H}_{P}^{\Theta}\Big)_1$ is the realification of the
transform $\mathcal{H}_{P}$, with the imaginary number replaced by the
matrix $i\gamma^0$ acting on the first Majorana index of $\psi$. 
\end{defn}

\begin{prop}
Consider a unitary operator $U:Pinor_j(\mathbb{R}^3)\to Pinor_j(\mathbb{X})$ such that
$U \circ H^2=E^2 \circ U$, where 
\begin{align*}
iH\{\Psi\}(\vec{x})\equiv
\Big(\gamma^0\vec{\slashed \partial}+i\gamma^0m\Big)_k\Psi(\vec{x})
\end{align*}
the Majorana matrices act on some Majorana index $k$; 
 $E^2\{\Phi\}(X)\equiv E^2(X)\Phi(X)$ with  $E(X)\geq m\geq 0$ a real
number. 

Then the operator $U':Pinor(\mathbb{R}^3)\to Pinor(\mathbb{X})$ is unitary,
where $U'$ is defined by:
\begin{align*}
U'\equiv \frac{E+U H\gamma^0 U^\dagger}{\sqrt{E+m}\sqrt{2E}}
\end{align*}
\end{prop}

\begin{proof}
Note that since $E^2=U^\dagger H^2 U$, $E=\sqrt{E^2}$ commutes with $U H\gamma^0
U^\dagger$. We have that
\begin{align*}
&(U')^\dagger (U')=\frac{E+U\gamma^0 H U^\dagger}{\sqrt{E+m}\sqrt{2E}}
\frac{E+U H\gamma^0 U^\dagger}{\sqrt{E+m}\sqrt{2E}}=1
\end{align*}
We also have that $(U')(U')^\dagger=1$. Therefore, $U'$ is unitary.
\end{proof}

\begin{defn}
The Fourier-Majorana transform 
$\mathcal{F}_M: Pinor_j(\mathbb{R}^3)\to Pinor_j(\mathbb{R}^3)$
is an unitary operator defined by:

\begin{align*}
\mathcal{F}_M\{\Psi\}(\vec{p})\equiv\int d^3\vec{x}
\Big(\frac{e^{-i\gamma^0\vec{p}\cdot\vec{x}}}{\sqrt{(2\pi)^3}}\Big)_1
\prod_{k=1}^{2j}\Big(\frac{E_p+H(\vec{x})\gamma^0}{\sqrt{E_p+m}\sqrt{2E_p}}\Big)_k
\Psi(\vec{x}),\ \Psi\in Pinor_j(\mathbb{R}^3)
\end{align*}
The matrices with the index $k$ apply on
the corresponding spinor index of $\Psi$.
\end{defn}

\begin{defn}
The Hankel-Majorana transform 
$\mathcal{H}_M: Pinor_j(\mathbb{R}^3)\to Pinor_j(\mathbb{S}^3)$
is an unitary operator defined by:
\begin{align*}
&\mathcal{H}_M\{\Psi\}(p,l,J,\nu)\equiv
\sum_{\mu=-l}^l\sum_{n=-j}^j <l\mu jn|J\nu>\int d^3\vec{x}\\
&\Big(\frac{2 p}{\sqrt{2\pi}}j_l(pr)Y_{l\mu}(\theta,\varphi)\Big)_1
\prod_{k=1}^{2j}\Big(\frac{E_p+H(\vec{x})\gamma^0}{\sqrt{E_p+m}\sqrt{2E_p}}\Big)_k
\Psi(\vec{x},n)
\end{align*}

The matrices with the index $k$ apply on
the corresponding spinor index of 
$\Psi\in Pinor_j(\mathbb{R}^3)$. 
$<l\mu jn|J\nu>$ are the Clebsh-Gordon coefficients and
$\Psi(\vec{x},n)\in Pinor_{j,n}$ such that
$\Psi(\vec{x})=\sum_{n=-j}^j \Psi(\vec{x},n)$.
\end{defn}

The inverse Fourier-Majorana transform verifies:
\begin{align*}
(iH(\vec{x}))_k\
\mathcal{F}_M^{-1}\{\psi\}(\vec{x})&=
(\mathcal{F}_M^{-1}\circ R)\{\psi\}(\vec{x})\\
\vec{\partial}_l\ \mathcal{F}_M^{-1}\{\psi\}(\vec{x})&=
(\mathcal{F}_M^{-1}\circ R')\{\psi\}(\vec{x})
\end{align*}
Where $\psi\in Pinor_j(\mathbb{R}^3)$ and
$R,R':Pinor_j(\mathbb{R}^3)\to Pinor_j(\mathbb{R}^3)$ are linear maps
defined by:
\begin{align*}
R\{\psi\}(\vec{p})&\equiv (i\gamma^0)_k E_p \psi(\vec{p})\\
R'\{\psi\}(\vec{p})&\equiv (i\gamma^0)_1\vec{p}_l\ \psi(\vec{p})
\end{align*}

The inverse Hankel-Majorana transform verifies:
\begin{align*}
(iH(\vec{x}))_k\
\mathcal{H}_M^{-1}\{\psi\}(\vec{x})&=
(\mathcal{H}_M^{-1}\circ R)\{\psi\}(\vec{x})\\
(-x^1\partial_2+x^2\partial_1+\sum_{k=1}^{2j}(i\gamma^0\gamma^3\gamma^5)_k)
\ \mathcal{H}_M^{-1}\{\psi\}(\vec{x})&=
(\mathcal{H}_M^{-1}\circ R')\{\psi\}(\vec{x})
\end{align*}
Where $\psi\in Pinor_j(\mathbb{S}^3)$ and $R,R':Pinor_j(\mathbb{S}^3)\to Pinor_j(\mathbb{S}^3)$ are
linear maps defined by:
\begin{align*}
R\{\psi\}(p,l,J,\nu)&\equiv (i\gamma^0)_k E_p \psi(p,l,J,\nu)\\
R'\{\psi\}(p,l,J,\nu)&\equiv (i\gamma^0)_1\nu\ \psi(p,l,J,\nu)
\end{align*}

\begin{defn}
The space of (real) Bargmann-Wigner fields $BW_{j}(\mathbb{R}^3)$ is defined as:
\begin{align*}
BW_j\equiv 
\{\Psi\in Pinor_{j}(\mathbb{R}^3):
\Big(e^{iH(\vec{x}) t}\Big)_k\Psi=\Big(e^{iH(\vec{x})t}\Big)_1 \Psi; 1\leq k\leq
2j; t\in \mathbb{R}\}
\end{align*}
\end{defn}

Note that if the equality $e^{-iH_1 t}\Psi=e^{-iH_2 t}\Psi$ holds for all differentiable
$\Psi\in H$ then for the continuous linear extension
the equality holds for all $\Psi\in H$, by the bounded linear transform theorem.

\begin{defn}
The complex Hilbert space 
$Dirac_j(\mathbb{X})\equiv Pinor_j(\mathbb{X})\otimes \mathbb{C}$
is the complexification of $Pinor_j(\mathbb{X})$.
The space of complex Bargmann-Wigner fields is the complexification of
the space of real Bargmann-Wigner fields.
\end{defn}

\subsection{Real unitary representations of the Poincare group}
\begin{defn}
The $IPin(3,1)$ group is defined as the semi-direct product
$Pin(3,1)\ltimes \mathbb{R}^4$, with the group's product defined as
$(A,a)(B,b)=(AB,a+\Lambda(A)b)$, for $A,B\in Pin(3,1)$ and 
$a,b\in \mathbb{R}^4$ and $\Lambda(A)$ is the Lorentz transformation
corresponding to $A$.

The $ISL(2,C)$ group is isomorphic to the subgroup of $IPin(3,1)$, 
obtained when $Pin(3,1)$ is restricted to $Spin^+(1,3)$. The full/restricted
Poincare group is the representation of the $IPin(3,1)/ISL(2,C)$ group on
Lorentz vectors, defined as 
$\{(\Lambda(A),a): A\in Pin(3,1), a\in \mathbb{R}^4\}$.
\end{defn}

\begin{defn}
Given a Lorentz vector $l$, the little group $G_l$ is the subgroup
of $SL(2,C)$ such that for all $g\in G_l$, $g\slashed l=\slashed l g$. 
\end{defn}

\begin{prop}
Given a Lorentz vector $l$, consider a set of matrices $\alpha_k\in
SL(2,C)$ verifying 
$\alpha_k \slashed l=\slashed k \alpha_k$. Let $H_k \equiv \{\alpha_{\Lambda_S(k)}^{-1} S\alpha_k: S\in SL(2,C)\}$. Then $H_k=G_l$. 
\end{prop}
\begin{proof}
We can check that $H_k\subset G_l$. For any $s\in G_l$, there is
$S=\alpha_{\Lambda_S(k)} s \alpha_k^{-1}$ such that $s\in H_k$. 
\end{proof}

For $i\slashed l=i\gamma^0$, we can set 
$\alpha_p=\frac{\slashed p\gamma^0+m}{\sqrt{E_p+m}\sqrt{2m}}$ and $G_l=SU(2)$.
For $i\slashed l=(i\gamma^0+i\gamma^3)$, we can set $\alpha_p=B_vR_{p}$,
where the boost velocity is $v=\frac{E_p^2-1}{E_p^2+1}$ along $\vec{p}$ and 
$R_p=e^{-\gamma^2\gamma^1 \theta/2}e^{-\gamma^1\gamma^3 \phi/2}$ is a rotation from the $z$
axis to the axis 
$\frac{\vec{\slashed p}}{E_p}=(\sin\phi \cos\theta \gamma_1+\sin\phi \sin\theta
\gamma_2+\cos\phi \gamma_3)$; $G_l=SE(2)$
\begin{align}
SE(2)=\{(1+i\gamma^5(\gamma^1a+\gamma^2b)(\gamma^0+\gamma^3))e^{i\gamma^0\gamma^3\gamma^5\theta}:
a,b,\theta\in \mathbb{R}\}.
\end{align}

\begin{rmk}
The complex irreducible projective representations of the Poincare
group with finite mass split into positive and negative energy
representations, 
which are complex conjugate of each other. They are labeled by one number $j$, with $2j$
being a natural number.
The positive energy representation spaces $V_j$ are, up to isomorphisms, written
as a symmetric tensor product of Dirac spinor
fields defined on the 3-momentum space, verifying
$(\gamma^0)_k\Psi_j(\vec{p})=\Psi_j(\vec{p})$. The matrices with the index $k$ apply in
the corresponding spinor index of $\Psi_j$.

The representation space $V_0$ is, up to isomorphisms, written in a
Majorana basis as a complex scalar defined on the 3-momentum space.

The representation map is given by:
\begin{align*}
L_S\{\Psi\}(\vec{p})&=\sqrt{\frac{(\Lambda^{-1})^0(p)}{E_p}}\prod_{k=1}^{2j}(\alpha^{-1}_{\Lambda(p)}S\alpha_p)_k\Psi(\vec{\Lambda}^{-1}(p))\\
T_a\{\Psi\}(\vec{p})&=e^{-i p\cdot a}\Psi(\vec{p})
\end{align*}
Where $\alpha_p=\frac{\slashed p\gamma^0+m}{\sqrt{E_p+m}\sqrt{2m}}$.
\end{rmk}

\begin{prop}
The real irreducible projective representations of the Poincare
group with finite mass are labeled by one number $j$, with $2j$
being a natural number.
The representation spaces $W_j$ are, up to isomorphisms, written
as a symmetric tensor product of Majorana spinor fields defined
on the 3-momentum space, verifying
$(i\gamma^0)_k\Psi_j(\vec{p})=(i\gamma^0)_1\Psi_j(\vec{p})$. 
The matrices with the index $k$ apply in the corresponding spinor index of
$\Psi_j$.

The representation space $V_0$ is, up to isomorphisms, written in a
Majorana basis as a real scalar defined on the 3-momentum space, times
the identity matrix of a Majorana spinor space.

The representation map is given by:
\begin{align*}
L_S\{\Psi\}(\vec{p})&=\sqrt{\frac{(\Lambda^{-1})^0(p)}{E_p}}\prod_{k=1}^{2j}(\alpha^{-1}_{\Lambda(p)}S\alpha_p)_k\Psi(\vec{\Lambda}^{-1}(p))\\
T_a\{\Psi\}(\vec{p})&=e^{-i\gamma^0 p\cdot a}\Psi(\vec{p})
\end{align*}
\end{prop}

\begin{rmk}
The complex irreducible projective representations of the Poincare
group with null mass and discrete helicity split into positive and negative energy
representations, 
which are complex conjugate of each other. They are labeled by one number $j$, with $2j$
being an integer number.
The positive energy representation spaces $V_j$ are, up to isomorphisms, written
as a symmetric tensor product of Dirac spinor
fields defined on the 3-momentum space, verifying
$(\gamma^0)_k\Psi_j(\vec{p})=\Psi_j(\vec{p})$ and $(\gamma^3\gamma^5)_k\Psi_j(\vec{p})=\pm \Psi_j(\vec{p})$,
with the plus sign if $j$ is positive and the minus sign if $j$ is
negative.

The representation space $V_0$ is, up to isomorphisms, written in a
Majorana basis as a scalar defined on the 3-momentum space.

The representation map is given by:
\begin{align*}
L_S\{\Psi\}(\vec{p})&=\sqrt{\frac{(\Lambda^{-1})^0(p)}{E_p}}\prod_{k=1}^{2j}(e^{i\gamma^0\gamma^3\gamma^5\theta})_k\Psi(\vec{\Lambda}^{-1}(p))\\
T_a\{\Psi\}(\vec{p})&=e^{-i p\cdot a}\Psi(\vec{p})
\end{align*}
Where $\theta$ is the angle of the rotation of the little group $SE(2)$.
\end{rmk}

\begin{rmk}
The real irreducible projective representations of the Poincare
group with null mass and discrete helicity are labeled by one number $j$, with $2j$
being an integer number.
The positive energy representation spaces $V_j$ are, up to isomorphisms, written
as a symmetric tensor product of Majorana spinor
fields defined on the 3-momentum space, verifying
$(i\gamma^0)_k\Psi_j(\vec{p})=(i\gamma^0)_1\Psi_j(\vec{p})$ and $(\gamma^3\gamma^5)_k\Psi_j(\vec{p})=\pm \Psi_j(\vec{p})$,
with the plus sign if $j$ is positive and the minus sign if $j$ is
negative.

The representation space $V_0$ is, up to isomorphisms, written in a
Majorana basis as the realification of the complex functions defined
on the 3-momentum space, with the operator correspondent to the
imaginary unit given by the matrix $i\gamma^0$ of a Majorana spinor space.

The representation map is given by:
\begin{align*}
L_S\{\Psi\}(\vec{p})&=\sqrt{\frac{(\Lambda^{-1})^0(p)}{E_p}}\prod_{k=1}^{2j}(e^{i\gamma^0\gamma^3\gamma^5\theta})_k\Psi(\vec{\Lambda}^{-1}(p))\\
T_a\{\Psi\}(\vec{p})&=e^{-i\gamma^0 p\cdot a}\Psi(\vec{p})
\end{align*}
Where $\theta$ is the angle of the rotation of the little group $SE(2)$.
\end{rmk}

\subsection{Localization}

The concept of a measure is essential in physics.
\begin{defn}[Measure]
A measure on a set $X$, is a function which assigns a non-negative real number
---the \emph{size}---to some subsets of $X$, such that:

1) the subsets which are assigned a size are called \emph{measurable} sets, 
the complement of a measurable set and the countable union of measurable sets are
measurable sets;

2) the size of the countable union of disjoint measurable sets is the sum of their
sizes. 
\end{defn}

\begin{defn}
Consider a measurable space $(X, M)$, where $M$ is a $\sigma$-algebra
of subsets of $X$. A projection-valued-measure, $\pi$, is a map from
$M$ to the set of self-adjoint projections on a Hilbert space $H$ such
that $\pi(X)$ is the identity operator on $H$ and the function
$<\psi,\pi(A)\psi>$, with $A\in M$ is a measure on $M$, for all
$\psi\in H$.
\end{defn}

\begin{defn}
Suppose now that $X$ is a representation of $G$. 
Then, a system of imprimitivity is a pair $(U,\pi)$, where $\pi$ is a
projection valued measure and $U$ an unitary representation of $G$ on
the Hilbert space $H$, such that $U(g)\pi(A) U^{-1}(g)=\pi(gA)$.
\end{defn}

\begin{rmk}[Theorem 6.12 of \cite{vara}] 
There is a one-to-one correspondence
between the complex system of imprimitivity (U,P), based on $\mathbb{R}^3$,
and the representations of $SU(2)$.
The system (U,P) is equivalent to the system
induced by the representation of $SU(2)$.
\end{rmk}

\begin{defn}
A covariant system of imprimitivity is a system of imprimitivity (U,P),
where $U$ is a representation of the Poincare group and $P$ is a projection-valued
measure based on $\mathbb{R}^3$, such that for the Euclidean group  
$U(g)\pi(A) U^{-1}(g)=\pi(gA)$ and for the Lorentz group, for a state at time null at 
point  $\vec{x}=0$, $L\{\Psi\}(0)=S\Psi(0)$.
\end{defn}

\begin{defn}
A localizable real unitary representation of the Poincare group, 
compatible with Poincare covariance, consists of a system of imprimitivity on $R^3$ for which
at time null and $\vec{x}=0$, the Lorentz transformations do not act on the space coordinates.
\end{defn}

So, the localization of a state in $x=0$ is a property 
invariant under relativistic transformations.

\begin{prop}
Any localizable unitary representation of the Poincare group,
compatible with Poincare covariance, is a direct sum of irreducible representations which are massive or massless with discrete helicity.
\end{prop}

\begin{proof}
Since the system is a unitary Poincare representation, it is a direct sum of irreducible unitary Poincare representations 
and so there must be an unitary transformation $U$, such that:
\begin{align}
&\Psi(x+a)=(U e^{-J P\cdot a}U^{-1})\{\Psi\}(x)\\
&S\Psi(\Lambda(x))=(U L U^{-1})\{\Psi\}(x)
\end{align}
Where $J$ is the operator corresponding to the imaginary unit after the 
realification of the Poincare representation, so $L$ commutes with $J$. 

The system of imprimitivity is a representation of $SU(2)$, hence the operator $i\gamma^0$ is well defined. 
If we make a Fourier transformation, then we get that:
\begin{align}
(U e^{J \vec{P}\cdot \vec{a}}U^{-1})\{\Psi\}(\vec{p})=e^{i\gamma^0 \vec{p}\cdot \vec{a}}\Psi(\vec{p})
\end{align}
Note that this equation is valid for all $\vec{p}$.
The system is a direct sum of irreducible unitary Poincare representations.
Then, for  $m^2<0$ only the subspace $\vec{p}^2\geq |m^2|$ is valid.
For $p=0$ only the subspace  $\vec{p}=0$ is valid. 
Since the other types of irreducible representations verify $p\neq 0$ and $m^2\geq 0$, 
the complementary subspaces $\vec{p}^2 < |m^2|$ or $\vec{p}\neq 0$ cannot be 
representation spaces and hence the representations with 
$m^2<0$ and $p=0$ cannot be subspaces of a localizable representation.

So we are left with $p\neq 0$ and $m^2\geq 0$. Now we can define a subspace for each $m^2$, such that the square of 
the generator of translations in time is given by $\vec{\partial}^2+m^2$. In each subspace there is a localizable representation.

Given a subspace with $p\neq 0$ and $m^2\geq 0$, 
$M$, we consider the subspace $N$ of the representation $M\oplus M_0$ verifying $e^{iH t}\Psi=e^{iH_0 t}\Psi$,
where $M_0$ is a spin-0 representation and $e^{iH_0 t}$ is the translation in time acting on $M_0$.
Then, $e^{iH_0(\vec{\partial}) t}U=Ue^{iH_0(J\vec{P})}$. Multiplying $U$ by 
$\alpha_p \sqrt{m/E_p}$ we can check that $J\Psi=i\gamma^0\Psi$ and so $N$ is equivalent to $M$.


Now we define the unitary transformation 
$\Lambda \{\Psi\}(p)=\sqrt{\frac{E_p}{\Lambda^0(p)}}\Psi(\Lambda^{-1}(p))$.
Then, we can check that $S\equiv L\Lambda^{-1}$ and it does not depend on $\vec{p}$.
If we redefine $U\{\Psi\}(\vec{p})= \alpha_p\sqrt{\frac{1}{\Lambda^0(p)}} U'\{ \Psi\}(\vec{p})$,
then we get that $\Lambda S \alpha_pU'\{\Psi\}(\vec{p})=\alpha_p\Lambda Q_pU'\{\Psi\}(\vec{p})$ 
and so $U'$ commutes with the Poincare representation.

If we look for subspaces where $m^2=0$ and the representation of $Q_p$ has infinite spin,
then the boost in the $z$ direction for a momenta in the $z$ direction 
multiplies the modulus of the translations of $SE(2)$ by $E_p$, 
which is in contradiction with the fact that
$S\equiv L\Lambda^{-1}$ does not depend on $\vec{p}$.

So, we are left with a direct sum of massive representations and massless with discrete helicity.
\end{proof}

\begin{prop}
For any complex localizable unitary representation of the Poincare group,
compatible with Poincare covariance, it if contains as a subspace a positive energy representation 
then it also contains the corresponding negative energy representation.
\end{prop}
\begin{proof}
The subspaces defined by the projectors involving the $i\gamma^0$s in the $Q_p$ representation are not conserved 
by the system of imprimitivity because $\gamma^0$ does not commute with the matrices $\vec{\gamma}\gamma^0$
present in the transformation from momenta to coordinate space. 
When we go back to coordinate space, the projector on the $i\gamma^0$s can be written as an equality of the time translations which is 
not part of the commuting ring of the SU(2) representation and hence it does not commute with the system of imprimitivity on $R^3$. 
\end{proof}

\begin{cor}
A localizable Poincare representation is an irreducible representation of the Poincare group (including parity) if and only if it is:
a)real and b)massive with spin 1/2 or massless with helicity 1/2.
\end{cor}
\begin{proof}
Since the subspaces defined by the projectors involving the $i\gamma^0$s in the $Q_p$ representation are not conserved 
by the system of imprimitivity, then the condition for irreducibility cannot involve such projectors, which only happens for real representations with
one spinor index. 
\end{proof}

Notice that the condition of irreducibility of the representations admits localized solutions---
the derivative of a bump function is a bump function, so we can find bump functions in the representation space---but it does not admit a position operator---the subspace of bump functions 
is not closed. Hence, we can say that a particular spin 1 state is in an arbitrarily small region of space, but the measurement of the position of an arbitrary spin 1 state might make it no longer a spin 1 state.

Going to complex systems, we can check that in the massive case, 
the condition of irreducibility does not  admit localized solutions---given a localized solution $\Psi$ in a region of space, then the result of the 
application of the projection operator to $\Psi$ is not localized in a region of space.
As for the massless representation, the condition of positive energy
does not admit localized solutions either---for the same region as above---, but the condition for a chiral irreducible representation 
does admit localized solutions. The parity operator for such a chiral irreducible representation is anti-linear.

The localizable Poincare representation is Poincare
covariant because for time $x^0=0$ at point $\vec{x}=0$, we have for the Lorentz
group $L\{\Psi\}(0)=S\Psi(0)$. The localizable Poincare representation is
compatible with causality because the propagator $\Delta(x)=0$ for $x^2<0$
(space-like $x$), where the propagator is defined for spin or helicity $1/2$ as:

\begin{align}
\Delta(x)\equiv \int \frac{d^3\vec{p}}{(2\pi)^32E_p}
\frac{\slashed p\gamma^0+m}{\sqrt{E_p+m}}e^{-i\gamma^0 p\cdot x}
\frac{\slashed p\gamma^0+m}{\sqrt{E_p+m}}
\end{align}
 And verifies:
\begin{align}
\Psi(x)=\int d^3\vec{y} \Delta(x-y)\Psi(y)
\end{align}
To show it we just need to do a Lorentz transformation such
that $x^0=0$ and then show that  $\Delta((0,\vec{x}))=0$ for
$\vec{x}\neq 0$.


\section{Energy Positivity}
\label{section:Energy}

\subsection{Vectors}
The role played by the unitary representations of the Poincare group in Quantum Theory corresponds to 
the role played by the unitary representations of the rotation group in Newton's mechanics.
Newton's mechanics is built using the 3-D vectors as its basic unit.
Some features of the non-commutative operators are already present in Newton's physics:
is a vector along the direction $z$? The answer can be yes, no, or part of it. 
And the result matters, for instance, the inner product of the 
vector Force and the vector of displacement is Work. 
The vector Force can come from someone pushing a block, from the gravitational or electric field or 
from friction. The Newton's vector is an abstraction to represent different  physical quantities,
all it matters are their properties as a unitary representation of the rotation group.
In the same way, a vector in Quantum Theory is a unitary 
representation of the Poincare group, it may be used to represent different physical objects.

\subsection{Density matrix and real Hilbert space}

As a consequence of Schur's lemma---related with the Frobenious theorem---, the set of normal operators commuting with an irreducible real unitary 
representation of a Lie group is isomorphic to the reals, to the complex numbers or to the quaternions---the 
irreduciblity of a group representation on a Hilbert space is intuitively the minimization of the degrees of freedom of the Hilbert space.
This fact turns the 
study of the Hilbert spaces over the reals, the complex or the quaternions interesting for Quantum Theory.
However, once we consider the density matrix in Quantum Mechanics, it is a simple exercise to show that the 
complex and quaternion Hilbert spaces are special cases of the real Hilbert space. 

In short, the complex Hilbert space case is achieved once we postulate that there is a 
unitary operator $J$, with $J^2=-1$, which commutes with the density matrix and all the observables.
The quaternionic Hilbert space corresponds to the case where both the unitary operators $J$ and $K$ commute with
the density matrix and all the observables, with $J^2=K^2=-1$ and $JK=-KJ$. Note that a complex 
Hilbert space is an Hilbert space over a division algebra over the real numbers, hence it has an extra layer of 
mathematical structure, which is dispensable because of the already existing density matrix in Quantum Mechanics. 

Of course, if the postulate corresponding to the complex Hilbert space is correct, 
there are practical advantages in using the complex notation. However, we should be aware that using the complex 
notation is a practical choice, not one of fundamental nature in the formalism of Quantum Mechanics. 
We cannot claim that the fact that the operator 
$J$ exists is a deductible consequence of the formalism of Quantum Mechanics with a complex Hilbert space.
It would be the same as claiming that we can derive from Newton's formalism that the
space is 3 dimensional, instead of assuming that we use 3 dimensional vectors in Newton 
mechanics because we postulate that the space has 3 dimensions.

\subsection{Localization}
The Hilbert space of non-relativistic Quantum Mechanics may be necessarily isomorphic to a
complex one\cite{yang,quaternionic,complexstructures,gibbons},
but there is no proper coordinate space---compatible with covariance under relativistic 
transformations---in the complex Hilbert space of an irreducible unitary representation of the Poincare group  
\cite{newton,localization,causality,stringfields,covariance}---the group of symmetries of the space-time. 
The complete definition of a Newton's 3-D vector includes the vector and 
the point in space where the vector is applied, but in a Quantum Mechanics' vector the application points are to be deduced from the vector 
itself; so it is of fundamental importance to be able to recover the coordinate space from the Hilbert space of a representation of 
the symmetries of space-time, otherwise the vector is of little meaning.

Choosing real representations is, in practice, choosing real Majorana spinors instead of
complex scalars as the basic elements of relativistic Quantum Theory.
For instance, we will see that the state of a spin-0 elementary system is a tensor 
field of real Majorana spinors, which only in momenta space (not in coordinate space) can
be considered a complex scalar field.
 
\subsection{Many particles}
In classical mechanics, the energy of a free body of mass $m$ is $E_p=\frac{\vec{p}^2}{2m}$. 
Since it is proportional to the square of the momentum, it does not make sense to talk about a negative energy. 
However, if we consider a box in which we can insert and remove free bodies such that in both the initial and final 
states the box is empty, the insertion of a body with momentum $\vec{p}$ and negative energy 
$E_p=-\frac{\vec{p}^2}{2m}$ to the system is equivalent to the removal of a body with momentum $-\vec{p}$ positive energy 
$E_p=\frac{\vec{p}^2}{2m}$, because the equations of motion are invariant under time reversal.
But time reversal transforms the act of adding a body on the act of removing a body. 
  
So, how can we say that a body was added to the system and not that the movie of the removal of a body is playing backwards?
The solution is to identify a feature on the system that is also affected by time reversal and we use it as a reference. 
For instance, if there is one body that---we know, or we define it as if---it was added to the system, 
then the addition of that body will appear a removal if we are watching the movie backwards. 
The product of the energies of two bodies is invariant under the Galilean transformations.
Note that we can only remove a body which was previously added to the box, 
as well as only add a body which will later be removed, to keep the box empty in both the initial and final states. 

Hence, the value of any quantity  which is non-invariant under the space-time symmetries---including the sign of the Energy---
by itself does not mean much without something to compare to, such that we can compute an invariant quantity.

In non-relativistic Quantum Mechanics, 
the translations in time
are given by the operator $e^{i\frac{\vec{\partial}^2}{2m}t}$---where $t$ is time---acting on a Hilbert space of \emph{positive} energy solutions 
because there is the imaginary unit---which is invariant under Lorentz transformations and anti-commutes with the time reversal transformations---
that we use as our reference.

In relativistic Quantum Mechanics, the translations in time are given by the operator
 $e^{(\gamma^0\vec{\gamma}\cdot \vec{\partial}+i\gamma^0 m)t}$, 
which is real---in the Majorana basis---and the position operator does \emph{not} leave invariant a Hilbert space of \emph{positive} 
Energy solutions. 
In other words, if we want a coordinate space which is relativistic covariant, the imaginary unit cannot be used as our reference for 
the sign of the energy. We cannot say that by considering real Hilbert spaces we are creating a new problem about Energy positivity.
as if we insist on a covariant coordinate space, the problem about the Energy positivity does not vanish in complex Hilbert spaces.
Remember that ever since the Dirac sea (which led to the prediction of the positron) the problem about Energy positivity 
was always solved in a many particle description.

In a system of particles, 
we can compare the energy of one particle with the energy of another particle we know it is positive,
like we would do in classical mechanics. 
If our reference particle is massive and has momentum $q$, then the Poincare
invariant condition $p\cdot q>0$ will be respected by a massive or massless particle 
with momentum $p$ if and only if $p^0$ has the same sign as $q^0$. 
Instead of the momenta we can use the translations generators to define the condition for energy positivity. 

\section{Conclusion}

The complex irreducible representations are not a generalization of the real
irreducible representations, in the same way that the complex numbers are a
generalization of the real numbers. There is a map, one-to-one or
two-to-one and surjective up to equivalence, from the complex to
the real irreducible representations of a Lie group on a Hilbert
space. 

We obtained all the real
unitary irreducible projective representations of the Poincare
group, with discrete spin, as real Bargmann-Wigner fields.
For each pair of complex representations with positive/negative
energy, there is one real representation.
The Majorana-Fourier and Majorana-Hankel unitary
transforms of the real Bargmann-Wigner fields relate the
coordinate space with the linear and angular momenta spaces.
The localizable unitary representations of the Poincare group 
(compatible with Poincare covariance and causality) are direct sums of 
irreducible representations with discrete spin and helicity, this result establishes a fundamental
difference between the representations associated to existing elementary systems and the other 
representations for which no existing elementary systems are known to be associated.

We might be interested in the position as an observable. Now the
question is, given an irreducible representation of the Poincare
group, should the position be invariant under a $U(1)$ symmetry? 
Unfortunately, everyone known to
the author that studied this problem assumed that it should. 
But the answer a priori is no it should not, because the $U(1)$ is
related with the gauge symmetry which is a local symmetry and it would
be useful to have a well defined notion of localization before we
start considering local symmetries. For the spin one-half in a real Hilbert space, 
the localization problems\cite{localization,causality,newton,wightman,vara,stringfields}
only appear  if we require that all the
observables are invariant under a U(1) symmetry---usually associated
with the charge---, this is related with the result in quantum field theory
that causality requires the existence of anti-particles\cite{weinberg}.

\bibliographystyle{utphys}
\bibliography{Poincare}

\providecommand{\href}[2]{#2}\begingroup\raggedright\begin{thebibliography}{10}

\bibitem{newton}
T.~D. Newton and E.~P. Wigner, ``Localized states for elementary systems,''
  \href{http://dx.doi.org/10.1103/RevModPhys.21.400}{{\em Rev. Mod. Phys.}
  {\bfseries 21} (Jul, 1949) 400--406}.

\bibitem{wignerquote}
E.~P. Wigner, ``The unreasonable effectiveness of mathematics in the natural
  sciences. richard courant lecture in mathematical sciences delivered at new
  york university, may 11, 1959,''
  \href{http://dx.doi.org/10.1002/cpa.3160130102}{{\em Communications on Pure
  and Applied Mathematics} {\bfseries 13} no.~1, (1960) 1--14}.
  \url{http://dx.doi.org/10.1002/cpa.3160130102}.

\bibitem{wigner}
E.~P. Wigner, ``{On Unitary Representations of the Inhomogeneous Lorentz
  Group},''
{\em Annals Math.} {\bfseries 40} (1939) 149--204.

\bibitem{mackey}
G.~Mackey, {\em Unitary group representations in physics, probability, and
  number theory}.
\newblock Advanced book classics. Addison-Wesley Pub. Co., 1989.

\bibitem{ohnuki}
Y.~Ohnuki, {\em Unitary Representations of the Poincar{\'e} Group and
  Relativistic Wave Equations}.
\newblock World Scientific Publishing Company Incorporated, 1988.

\bibitem{poincare}
N.~{Straumann}, ``{Unitary Representations of the inhomogeneous Lorentz Group
  and their Significance in Quantum Physics},'' {\em ArXiv e-prints} (Sept.,
  2008) , \href{http://arxiv.org/abs/0809.4942}{{\ttfamily arXiv:0809.4942
  [math-ph]}}.

\bibitem{weinberg}
S.~Weinberg, {\em The Quantum Theory of Fields: Modern Applications}.
\newblock The Quantum Theory of Fields. Cambridge University Press, 1995.

\bibitem{knapp}
A.~Knapp, {\em Representation Theory of Semisimple Groups: An Overview Based on
  Examples}.
\newblock Princeton Landmarks in mathematics and physics. Princenton University
  Press, 2001.

\bibitem{feynmanrules}
S.~{Weinberg}, ``{Feynman Rules for Any Spin},''
  \href{http://dx.doi.org/10.1103/PhysRev.133.B1318}{{\em Physical Review}
  {\bfseries 133} (Mar., 1964) 1318--1332}.

\bibitem{symmetry}
D.~Gross, ``{Symmetry in physics: Wigner's legacy},''
{\em Phys.Today} {\bfseries 48N12} (1995) 46--50.

\bibitem{realQM}
J.~Myrheim, ``{Quantum mechanics on a real hilbert space},'' {\em ArXiv
  e-prints} (1999) ,
\href{http://arxiv.org/abs/quant-ph/9905037}{{\ttfamily arXiv:quant-ph/9905037
  [quant-ph]}}.

\bibitem{realqft}
E.~Stueckelberg, ``Quantum theory in real hilbert-space,''
  \href{http://dx.doi.org/10.5169/seals-113093}{{\em Helvetica Physica Acta}
  {\bfseries 33} (1960) 727--752}.

\bibitem{realqftII}
E.~Stueckelberg and M.~Guenin, ``Quantum theory in real hilbert-space. ii,
  addenda and errats,'' \href{http://dx.doi.org/10.5169/seals-113188}{{\em
  Helvetica Physica Acta} {\bfseries 34} (1961) 621--628}.

\bibitem{realqftIII}
E.~Stueckelberg, M.~Guenin, and C.~Piron, ``Quantum theory in real
  hilbert-space. iii, fields of the first kind (linear field operators),''
  \href{http://dx.doi.org/10.5169/seals-113192}{{\em Helvetica Physica Acta}
  {\bfseries 34} (1961) 621--628}.

\bibitem{quantumstatistics}
L.~Accardi and A.~Fedullo, ``On the statistical meaning of complex numbers in
  quantum mechanics,'' \href{http://dx.doi.org/10.1007/BF02817051}{{\em Lettere
  al Nuovo Cimento} {\bfseries 34} no.~7, (1982) 161--172}.

\bibitem{hestenes_old}
D.~Hestenes, ``{Real Spinor Fields},''
  \href{http://dx.doi.org/10.1063/1.1705279}{{\em J. Math. Phys.} {\bfseries 8}
  no.~4, (1967) 798--808}.

\bibitem{classicalfields}
N.~J. {Poplawski}, ``{Spacetime and fields},'' {\em ArXiv e-prints} (Nov.,
  2009) , \href{http://arxiv.org/abs/0911.0334}{{\ttfamily arXiv:0911.0334
  [gr-qc]}}.

\bibitem{gravitypoincare}
G.~Grignani and G.~Nardelli, ``Gravity and the poincaré group,''
  \href{http://dx.doi.org/10.1103/PhysRevD.45.2719}{{\em Phys. Rev. D}
  {\bfseries 45} (Apr, 1992) 2719--2731}.

\bibitem{gravitypoincare2}
F.~W. Hehl, Y.~N. Obukhov, and D.~Puetzfeld, ``{On Poincar\'e gauge theory of
  gravity, its equations of motion, and Gravity Probe B},''
  \href{http://dx.doi.org/10.1016/j.physleta.2013.04.055}{{\em Phys.Lett.}
  {\bfseries A377} (2013) 1775--1781},
\href{http://arxiv.org/abs/1304.2769}{{\ttfamily arXiv:1304.2769 [gr-qc]}}.

\bibitem{brancocp}
G.~Branco, L.~Lavoura, and J.~Silva, {\em CP Violation}.
\newblock International series of monographs on physics. Clarendon Press, 1999.

\bibitem{strongcp}
R.~D. {Peccei}, ``{The Strong CP Problem and Axions},'' in {\em Axions},
  M.~{Kuster}, G.~{Raffelt}, and B.~{Beltr{\'a}n}, eds., vol.~741 of {\em
  Lecture Notes in Physics, Berlin Springer Verlag}, pp.~3--540.
\newblock 2008.
\newblock \href{http://arxiv.org/abs/arXiv:hep-ph/0607268}{{\ttfamily
  arXiv:hep-ph/0607268}}.

\bibitem{imbalance}
L.~{Canetti}, M.~{Drewes}, and M.~{Shaposhnikov}, ``{Matter and antimatter in
  the universe},'' \href{http://dx.doi.org/10.1088/1367-2630/14/9/095012}{{\em
  New Journal of Physics} {\bfseries 14} no.~9, (Sept., 2012) 095012},
  \href{http://arxiv.org/abs/1204.4186}{{\ttfamily arXiv:1204.4186 [hep-ph]}}.

\bibitem{localization}
J.~F. Perez and I.~F. Wilde, ``Localization and causality in relativistic
  quantum mechanics,'' \href{http://dx.doi.org/10.1103/PhysRevD.16.315}{{\em
  Phys. Rev. D} {\bfseries 16} (Jul, 1977) 315--317}.

\bibitem{causality}
G.~C. {Hegerfeldt}, ``{Instantaneous spreading and Einstein causality in
  quantum theory},''
  \href{http://dx.doi.org/10.1002/(SICI)1521-3889(199812)7:7/8<716::AID-ANDP716>3.0.CO;2-T}{{\em
  Annalen der Physik} {\bfseries 510} (Dec., 1998) 716--725},
  \href{http://arxiv.org/abs/quant-ph/9809030}{{\ttfamily quant-ph/9809030}}.

\bibitem{stringfields}
M.~{Plaschke} and J.~{Yngvason}, ``{Massless, string localized quantum fields
  for any helicity},'' \href{http://dx.doi.org/10.1063/1.3700765}{{\em Journal
  of Mathematical Physics} {\bfseries 53} no.~4, (Apr., 2012) 042301},
  \href{http://arxiv.org/abs/1111.5164}{{\ttfamily arXiv:1111.5164 [math-ph]}}.
  section 2.4 on the self-adjoint fields.

\bibitem{wightman}
A.~Wightman, ``{On the localizability of quantum mechanical systems},''
\href{http://dx.doi.org/10.1103/RevModPhys.34.845}{{\em Rev.Mod.Phys.}
  {\bfseries 34} (1962) 845--872}.

\bibitem{vara}
V.~Varadarajan, {\em Geometry of Quantum Theory}.
\newblock Springer Science+Business Media LLC, 2007.

\bibitem{diracsea}
{Alvarez-Gaume, Luis and Vazquez-Mozo, Miguel A.}, ``{Introductory lectures on
  quantum field theory},''
\href{http://arxiv.org/abs/hep-th/0510040}{{\ttfamily arXiv:hep-th/0510040
  [hep-th]}}.

\bibitem{mathQM}
G.~Teschl, {\em Mathematical Methods in Quantum Mechanics: With Applications to
  Schr{\"o}dinger Operators}.
\newblock Graduate studies in mathematics. American Mathematical Society, 2009.

\bibitem{Hall}
B.~Hall, {\em Lie Groups, Lie Algebras, and Representations: An Elementary
  Introduction}.
\newblock Graduate Texts in Mathematics. Springer, 2003.

\bibitem{locallycompact}
R.~W. Henrichs, ``On decomposition theory for unitary representations of
  locally compact groups,''
  \href{http://dx.doi.org/10.1016/0022-1236(79)90099-5}{{\em Journal of
  Functional Analysis} {\bfseries 31} no.~1, (1979) 101 -- 114}.

\bibitem{realalgebras}
A.~Oni{\v{s}}{\v{c}}ik and E.~S. I.~I. for Mathematical~Physics, {\em Lectures
  On Real Semisimple Lie Algebras And Their Representations}.
\newblock ESI Lectures in Mathematics and Physics. European Mathematical
  Society Publishing, 2004.

\bibitem{realirrep}
N.~Iwahori, ``{On real irreducible representations of Lie algebras.},'' {\em
  Nagoya Math. J.} {\bfseries 14} (1959) 59--83.

\bibitem{schur}
D.~Ramakrishnan and R.~Valenza, {\em Fourier Analysis on Number Fields}.
\newblock Graduate Texts in Mathematics. Springer, 1999.

\bibitem{compactlie}
T.~Br{\"o}cker and T.~Dieck, {\em Representations of Compact Lie Groups}.
\newblock Graduate Texts in Mathematics. Springer, 1985.

\bibitem{spinorsrealhilbert}
R.~Plymen and P.~Robinson, {\em Spinors in Hilbert Space}.
\newblock Cambridge Tracts in Mathematics. Cambridge University Press, 1994.

\bibitem{pin}
M.~{Berg}, C.~{De Witt-Morette}, S.~{Gwo}, and E.~{Kramer}, ``{The Pin Groups
  in Physics},'' \href{http://dx.doi.org/10.1142/S0129055X01000922}{{\em
  Reviews in Mathematical Physics} {\bfseries 13} (2001) 953--1034},
  \href{http://arxiv.org/abs/math-ph/0012006}{{\ttfamily math-ph/0012006}}.

\bibitem{todorov}
I.~{Todorov}, ``{Clifford Algebras and Spinors},'' {\em ArXiv e-prints} (June,
  2011) , \href{http://arxiv.org/abs/1106.3197}{{\ttfamily arXiv:1106.3197
  [math-ph]}}.

\bibitem{irreducible}
A.~Aste, ``{A direct road to Majorana fields},''
  \href{http://dx.doi.org/10.3390/sym2041776}{{\em Symmetry} {\bfseries 2}
  (2010) 1776--1809}, \href{http://arxiv.org/abs/0806.1690}{{\ttfamily
  arXiv:0806.1690 [hep-th]}}.
See section 5 on the Majorana spinor irrep of SL(2,C).

\bibitem{pal}
P.~B. {Pal}, ``{Dirac, Majorana, and Weyl fermions},''
  \href{http://dx.doi.org/10.1119/1.3549729}{{\em American Journal of Physics}
  {\bfseries 79} (May, 2011) 485--498},
  \href{http://arxiv.org/abs/1006.1718}{{\ttfamily arXiv:1006.1718 [hep-ph]}}.

\bibitem{dreiner}
H.~K. Dreiner, H.~E. Haber, and S.~P. Martin, ``{Two-component spinor
  techniques and Feynman rules for quantum field theory and supersymmetry},''
  \href{http://dx.doi.org/10.1016/j.physrep.2010.05.002}{{\em Phys.Rept.}
  {\bfseries 494} (2010) 1--196},
\href{http://arxiv.org/abs/0812.1594}{{\ttfamily arXiv:0812.1594 [hep-ph]}}.

\bibitem{wignertheorem}
L.~Moln\'{a}r, ``An algebraic approach to wigner's unitary-antiunitary
  theorem,'' \href{http://dx.doi.org/10.1017/S144678870003593X}{{\em Journal of
  the Australian Mathematical Society (Series A)} {\bfseries 65} (12, 1998)
  354--369}.

\bibitem{Dirac}
P.~A.~M. Dirac, ``The quantum theory of the electron,'' {\em Proc. R. Soc.
  Lond. A} {\bfseries 117} no.~778, (1928) 610--624.

\bibitem{BW}
V.~Bargmann and E.~P. Wigner, ``Group theoretical discussion of relativistic
  wave equations,'' {\em Proceedings of the National Academy of Sciences}
  {\bfseries 34} no.~5, (1948) 211--223,
  \href{http://arxiv.org/abs/http://www.pnas.org/content/34/5/211.full.pdf+html}{{\ttfamily
  http://www.pnas.org/content/34/5/211.full.pdf+html}}.

\bibitem{allspins}
S.-Z. Huang, T.-N. Ruan, N.~Wu, and Z.-P. Zheng, ``{Wavefunctions for particles
  with arbitrary spin},''
{\em Commun.Theor.Phys.} {\bfseries 37} (2002) 63--74.

\bibitem{revfoldy}
J.~P. Costella and B.~H. McKellar, ``{The Foldy-Wouthuysen transformation},''
  \href{http://dx.doi.org/10.1119/1.18017}{{\em Am.J.Phys.} {\bfseries 63}
  (1995) 1119},
\href{http://arxiv.org/abs/hep-ph/9503416}{{\ttfamily arXiv:hep-ph/9503416
  [hep-ph]}}.

\bibitem{foldy}
L.~L. Foldy and S.~A. Wouthuysen, ``On the dirac theory of spin 1/2 particles
  and its non-relativistic limit,''
  \href{http://dx.doi.org/10.1103/PhysRev.78.29}{{\em Phys. Rev.} {\bfseries
  78} (Apr, 1950) 29--36}.

\bibitem{squareroot}
E.~Hitzer, J.~Helmstetter, and R.~Ab{\l}amowicz,
  \href{http://dx.doi.org/10.1007/978-3-0348-0603-9_7}{``Square roots of -1 in
  real clifford algebras,''} in {\em Quaternion and Clifford Fourier Transforms
  and Wavelets}, E.~Hitzer and S.~J. Sangwine, eds., Trends in Mathematics,
  pp.~123--153.
\newblock Springer Basel, 2013.

\bibitem{clifford}
H.~De~Bie, ``Clifford algebras, fourier transforms, and quantum mechanics,''
  \href{http://dx.doi.org/10.1002/mma.2679}{{\em Mathematical Methods in the
  Applied Sciences} {\bfseries 35} no.~18, (2012) 2198--2228}.

\bibitem{image}
T.~Batard, M.~Berthier, and C.~Saint-Jean,
  \href{http://dx.doi.org/10.1007/978-1-84996-108-0_8}{``Clifford-fourier
  transform for color image processing,''} in {\em Geometric Algebra
  Computing}, E.~Bayro-Corrochano and G.~Scheuermann, eds., pp.~135--162.
\newblock Springer London, 2010.

\bibitem{negativeprobs}
P.~A.~M. {Dirac}, ``{Bakerian Lecture. The Physical Interpretation of Quantum
  Mechanics},'' \href{http://dx.doi.org/10.1098/rspa.1942.0023}{{\em Royal
  Society of London Proceedings Series A} {\bfseries 180} (Mar., 1942) 1--40}.

\bibitem{desitter}
S.~De~Bi\`evre and J.~Renaud, ``Massless gupta-bleuler vacuum on the
  (1+1)-dimensional de sitter space-time,''
  \href{http://dx.doi.org/10.1103/PhysRevD.57.6230}{{\em Phys. Rev. D}
  {\bfseries 57} (May, 1998) 6230--6241}.

\bibitem{krein}
B.~Forghan, M.~Takook, and A.~Zarei, ``{Krein Regularization of QED},''
  \href{http://dx.doi.org/10.1016/j.aop.2012.06.003}{{\em Annals Phys.}
  {\bfseries 327} (2012) 2388--2401},
\href{http://arxiv.org/abs/1206.2796}{{\ttfamily arXiv:1206.2796 [hep-ph]}}.

\bibitem{representations}
A.~Kirillov and E.~Hewitt, {\em Elements of the Theory of Representations}.
\newblock Grundlehren Der Mathematischen Wissenschaften. Springer London,
  Limited, 2011.

\bibitem{commutingring}
H.-D. {Doebner}, P.~{{\v S}{\v t}ov{\'{\i}}{\v c}ek}, and J.~{Tolar},
  ``{Quantization of Kinematics on Configuration Manifolds},''
  \href{http://dx.doi.org/10.1142/S0129055X0100079X}{{\em Reviews in
  Mathematical Physics} {\bfseries 13} (2001) 799--845},
  \href{http://arxiv.org/abs/math-ph/0104013}{{\ttfamily math-ph/0104013}}.

\bibitem{squareimprimitivity}
G.~Cassinelli, E.~De~Vito, and A.~Levrero, ``Square-integrable imprimitivity
  systems,'' \href{http://dx.doi.org/10.1063/1.533382}{{\em Journal of
  Mathematical Physics} {\bfseries 41} no.~7, (2000) 4833--4859}.

\bibitem{diracmatrices}
R.~H. {Good}, ``{Properties of the Dirac Matrices},''
  \href{http://dx.doi.org/10.1103/RevModPhys.27.187}{{\em Reviews of Modern
  Physics} {\bfseries 27} (Apr., 1955) 187--211}.

\bibitem{realgamma}
L.~O'Raifeartaigh, ``The dirac matrices and the signature of the metric
  tensor,'' \href{http://dx.doi.org/10.5169/seals-113184}{{\em Helvetica
  Physica Acta} {\bfseries 34} (1961) 675--698}.

\bibitem{bessel}
G.~S. {Adkins}, ``{Three-dimensional Fourier transforms, integrals of spherical
  Bessel functions, and novel delta function identities},'' {\em ArXiv
  e-prints} (Feb., 2013) , \href{http://arxiv.org/abs/1302.1830}{{\ttfamily
  arXiv:1302.1830 [math-ph]}}.

\bibitem{harmonics}
R.~Szmytkowski, ``Recurrence and differential relations for spherical
  spinors,'' \href{http://dx.doi.org/10.1007/s10910-006-9110-0}{{\em Journal of
  Mathematical Chemistry} {\bfseries 42} no.~3, (2007) 397--413}.

\bibitem{yang}
C.~N. Yang, \href{http://dx.doi.org/10.1017/CBO9780511564253}{``{Square root of
  minus one, complex phases and Erwin Schr\"{o}dinger},''} in {\em
  {Schr\"{o}dinger}}, C.~W. Kilmister, ed.
\newblock Cambridge University Press, 1987.

\bibitem{quaternionic}
S.~Adler, {\em Quaternionic Quantum Mechanics and Quantum Fields}.
\newblock Oxford University Press, USA, 1995.

\bibitem{complexstructures}
A.~{Trautman}, ``{On Complex Structures in Physics},'' in {\em On Einstein's
  Path: essays in honor of Engelbert Schucking}, A.~{Harvey}, ed., p.~487.
\newblock 1999.
\newblock \href{http://arxiv.org/abs/math-ph/9809022}{{\ttfamily
  math-ph/9809022}}.

\bibitem{gibbons}
G.~W. {Gibbons}, ``{The emergent nature of time and the complex numbers in
  quantum cosmology},'' {\em ArXiv e-prints} (Nov., 2011) ,
  \href{http://arxiv.org/abs/1111.0457}{{\ttfamily arXiv:1111.0457 [gr-qc]}}.

\bibitem{covariance}
S.~Farkas, Z.~Kurucz, and M.~Weiner, ``Poincaré covariance of relativistic
  quantum position,'' \href{http://dx.doi.org/10.1023/A:1013221616586}{{\em
  International Journal of Theoretical Physics} {\bfseries 41} no.~1, (2002)
  79--88}.

\end{thebibliography}\endgroup
\end{document}